\documentclass[leqno,10pt]{amsart}

\usepackage[full]{textcomp}
\usepackage{newpxtext}
\usepackage{times}
\usepackage{cabin} 
\usepackage[varqu,varl]{inconsolata} 
\usepackage[bigdelims,vvarbb]{newpxmath} 
\usepackage[cal=cm,bb=esstix,bbscaled=1.05]{mathalfa} 
\usepackage[bottom=.96in,top=.96in,left=1.25in,right=1.25in]{geometry}
\usepackage{savesym}
\let\savedegree\bigtimes
\let\bigtimes\relax
\usepackage{mathabx}
\let\bigtimes\savedegree
\usepackage[usenames,dvipsnames]{color}
\usepackage{hyperref}
\hypersetup{
 colorlinks,
 linkcolor={blue!90!black},
 citecolor={red!80!black},
 urlcolor={blue!50!black}
}
\usepackage{tikz}
\usetikzlibrary{decorations.markings}
\usepackage{mathrsfs}
\usepackage[all,cmtip]{xy}
\usepackage{mleftright}
\usepackage{booktabs}
\usepackage{cite}
\usepackage{enumitem}

\usepackage{graphicx}
\usepackage{subfigure}
\usepackage{caption}
\usepackage{placeins}
\usepackage{comment}
\usepackage[normalem]{ulem} 

\newcommand{\Hilbert}{H}
\newcommand{\Projection}{P}

\newcommand{\la}{\langle}
\newcommand{\ra}{\rangle}
\newcommand{\sgn}{\operatorname{sgn}}
\newcommand{\bds}{\boldsymbol}
\newcommand{\PV}{\operatorname{PV}}
\newcommand{\ph}{\text{phys}}
\newcommand{\bbb}{\sigma}
\allowdisplaybreaks
\usepackage{setspace}

\setlist[enumerate]{labelsep=*, leftmargin=1.5pc}
\setlist[enumerate]{label=\normalfont(\roman*), ref=\roman*}

\newtheorem{theorem}{Theorem}[section]

\newtheorem{cor}[theorem]{Corollary}

\theoremstyle{definition}
\newtheorem{example}[theorem]{Example}
\newtheorem{remark}[theorem]{Remark}
\newtheorem{definition}[theorem]{Definition}

\numberwithin{equation}{section}

\newcommand{\IGNORE}[1]{}
\newcommand{\ignore}[1]{}
\newcommand{\veps}{\varepsilon}
\newcommand{\opn}{\operatorname}

\newcommand{\im}{\operatorname{Im}}

\newcommand{\phm}{\phantom{-}}
\newcommand{\mbb}[1]{\mathbb{#1}}

\newcommand{\mc}[1]{\mathcal{#1}}

\newcommand{\jd}{\displaystyle}

\newcommand{\jt}{\textstyle}

\newcommand{\pa}{\partial}
\newcommand{\wtil}{\widetilde}

\usepackage{xcolor}
\usepackage{etoolbox}
\makeatletter
\pretocmd\@bibitem{\color{black}\csname keycolor#1\endcsname}{}{\fail}
\newcommand\citecolor[1]{\@namedef{keycolor#1}{\color{blue}}}
\makeatother

\begin{document}

\author[Jon Wilkening and Xinyu Zhao]{Jon Wilkening and Xinyu Zhao}
\address{Department of Mathematics\\ University of California at
  Berkeley\\Berkeley, CA\\94720\\USA}
\email{wilkening@berkeley.edu}
\email{zhaoxinyu@berkeley.edu}
\thanks{This work was supported in part by the National Science
  Foundation under award number DMS-1716560 and by the Department of
  Energy, Office of Science, Applied Scientific Computing Research,
  under award number DE-AC02-05CH11231.}
\subjclass[]{}
\title{Spatially quasi-periodic water waves of infinite depth}

\begin{abstract}
  We formulate the two-dimensional gravity-capillary water wave
  equations in a spatially quasi-periodic setting and present a
  numerical study of solutions of the initial value problem.  We
    propose a Fourier pseudo-spectral discretization of the equations
    of motion in which one-dimensional quasi-periodic functions are
    represented by two-dimensional periodic functions on a torus.  We
  adopt a conformal mapping formulation and employ a quasi-periodic
  version of the Hilbert transform to determine the normal velocity of
  the free surface. Two methods of time-stepping the initial
  value problem are proposed, an explicit Runge-Kutta (ERK)
  method and an exponential time-differencing (ETD) scheme. The
  ETD approach makes use of the small-scale decomposition to eliminate
  stiffness due to surface tension.  We perform a convergence study to
  compare the accuracy and efficiency of the methods on a traveling
  wave test problem.  We also present an example of a periodic wave
  profile containing vertical tangent lines that is set in motion with
  a quasi-periodic velocity potential.  As time evolves, each
  wave peak evolves differently, and only some of them
  overturn.  Beyond water waves, we argue that spatial
  quasi-periodicity is a natural setting to study the dynamics of
  linear and nonlinear waves, offering a third option to the usual
  modeling assumption that solutions either evolve on a periodic
  domain or decay at infinity.
\end{abstract}

\maketitle \markboth{J. WILKENING AND X. ZHAO}{SPATIALLY
  QUASI-PERIODIC WATER WAVES}

\section{Introduction}

Linear and nonlinear wave equations are generally studied under the
assumption that the solution is spatially periodic or decays to zero
at infinity \cite{lax:76}. Beginning with Berenger \cite{berenger}, a
great deal of effort has been devoted to developing perfectly matched
layer (PML) techniques for imposing absorbing boundary conditions over
a finite computational domain to simulate wave propagation problems on
unbounded domains. However, in many situations, assuming the waves
decay to zero at infinity is not a realistic model. For example, a
large body of water such as the ocean is often covered in surface
waves in every direction over vast distances.  But assuming spatial
periodicity may limit one's ability to observe interesting
dynamics. In this paper, we formulate the initial value problem of the
surface water wave equations in a spatially quasi-periodic setting,
design numerical algorithms to compute such waves, and study their
properties.

Since the pioneering work of Benjamin and Feir
  \cite{benj:feir:67} and Zakharov \cite{zakharov1968stability}, it
  has been recognized that water waves exhibit interesting nonlinear
  interactions between component waves of different wavelength.  For
  example, in oceanography, modulational instabilities of periodic
  wavetrains introduce perturbations that lead to spatially
  quasi-periodic dynamics and are believed to be one of the mechanisms
  responsible for the formation of rogue waves \cite{osborne2000,
    onorato2006modulational, ablowitz2015interacting}. These
  instabilities have been studied extensively using a variety of
  techniques, summarized below, including linearization using Bloch
  stability theory, evolving the nonlinear equations on a larger
  periodic domain, developing coupled weakly nonlinear models, and
  solving weakly nonlinear models via the inverse scattering
  transform. However, it has not been known how to formulate or
  compute fully nonlinear water waves in a spatially quasi-periodic
  setting. We show that a conformal mapping formulation of the water
  wave equations, introduced by Dyachenko \emph{et al.}
  \cite{dyachenko1996analytical} and further developed by many authors
  \cite{dyachenko1996nonlinear, choi1999exact, dyachenko2001dynamics,
    zakharov2002new, li2004numerical, milewski:10, wang2013two,
    viotti2014conformal, dyachenko:newell:16, turner:bridges}, extends
  nicely to this setting via a quasi-periodic generalization of the
  Hilbert transform. Currently our method is limited to
  two-dimensional fluids, but we formulate the equations of motion and
  discuss computational challenges of 3D quasi-periodic water waves in
  Appendix~\ref{sec:3d}.

The Bloch stability approach can be carried out by linearizing
  the full water wave equations about a traveling Stokes wave
  \cite{longuet:78, mclean:82, mackay:86, oliveras:11,
    trichtchenko:16, waterTS2} or within a weakly nonlinear model such
  as the nonlinear Schr\"odinger (NLS) equation
  \cite{benney:newell:67, zakharov1968stability}.  A major drawback is
  that unstable modes grow exponentially forever and eventually leave
  the realm of validity of the linearization. Osborne
  et.~al.~\cite{osborne2000} have observed that if nonlinear effects
  are taken into account in this scenario, the perturbation often
  exhibits Fermi-Pasta-Ulam recurrence \cite{berman:FPU:2005}. They
  solve a 2+1-dimensional NLS equation on a domain that is 10 times
  larger than the wavelength of the carrier wave and look for
  rogue-wave formation and recurrence over long simulation
  times. Similarly, Bryant and Stiassnie \cite{bryant:stiassnie:94}
  study recurrence using both a weakly nonlinear model (Zakharov's
    equation) and the full water wave equations in the context of
  standing water waves when the wavelength of the subharmonic
  perturbation is 9 times that of the unperturbed standing wave. The
  main drawback of this approach is that the larger periodic
  computational domain must be an integer multiple of both the base
  wave and the perturbation, which requires that the ratio of their
  wavelengths be a rational number with a small numerator and
  denominator. We propose a method below that allows for more general
  perturbations and plan to investigate the long-time nonlinear
  dynamics of unstable subharmonic perturbations of traveling and
  standing waves in future work.

An alternative approach that does not require rationally related
  wave numbers is to model the interaction of two periodic wavetrains
  as a coupled weakly nonlinear system.  This is particularly
  useful for studying the interaction between oblique waves on the
  surface of a three-dimensional fluid. For example, Bridges and
  Laine-Pearson studied coupled NLS equations \cite{bridges2001} and
  extended the theory to analyze the stability of short-crested waves
  \cite{bridges2005}. More recently, Ablowitz and Horikis
  \cite{ablowitz2015interacting} showed that some propagation angles
  enhance the number and amplitude of rogue wave events in a coupled
  NLS water wave model.

Some weakly nonlinear models are completely integrable and can
  be studied using the inverse scattering transform (IST)
  \cite{ablowitz:segur}.  Osborne et.~al.~discuss IST results for the
  1+1-dimensional NLS equation in the context of rogue waves in
  \cite{osborne2000}.  Other equations such as the Korteweg-deVries
  and Benjamin-Ono equations are meant to model wave dynamics in
  shallow water \cite{ablowitz:segur} and internal waves in a
  stratified fluid \cite{ono75}, respectively. Solving these equations
  using the inverse scattering transform \cite{ablowitz:segur} leads to
  infinite hierarchies of exact spatio-temporal quasi-periodic
  solutions \cite{flaschka,dobro:91}.

In these and other examples involving weakly nonlinear theory,
  it is natural to seek analogous quasi-periodic solutions of the
  Euler equations in regimes where the model equations are intended to
  be accurate. It is also of interest to search for new regimes and
  behavior not predicted by model water wave equations. Even within
  weakly nonlinear theory, except for the exact quasi-periodic
  solutions obtained via the IST, spatial quasi-periodicity has only
  been approximated by embedding in a larger periodic domain or by
  introducing coupling terms between two or more single-mode NLS
  equations. The framework we propose below for water waves, which
  involves representing quasi-periodic functions as periodic functions
  on a higher dimensional torus and using a spectral method to solve a
  torus version of the equations of motion, could also be used to find
  true quasi-periodic solutions of weakly nonlinear equations without
  introducing systems of coupled equations.

Only recently have quasi-periodic dynamics of water waves been
  studied mathematically. Berti and Montalto \cite{berti2016quasi} and
  Baldi et.~al.~\cite{baldi2018time} used Nash-Moser theory to prove
  the existence of small-amplitude temporally quasi-periodic
  gravity-capillary standing waves.  With different assumptions on the
  form of solutions, Berti~et.~al.~\cite{berti2020traveling} have
  proved the existence of time quasi-periodic gravity-capillary waves
  with constant vorticity while Feola and Giuliani
  \cite{feola2020trav} have proved the existence of time
  quasi-periodic irrotational gravity waves.  New families of
relative-periodic \cite{collision} and traveling-standing
\cite{waterTS} water wave solutions have been computed by Wilkening.
As with \cite{berti2016quasi, baldi2018time, berti2020traveling,
  feola2020trav}, these solutions are quasi-periodic in time rather
than space.

Another motivation for studying spatially quasi-periodic water waves
is the work of Wilton \cite{wilton1915lxxii}, who observed that a
resonance can occur that causes Stokes' regular perturbation expansion
for traveling water waves to break down \cite{akers2012wilton,
  akers2020wilton,
  vandenBroeck:book, trichtchenko:16, akers2019periodic}. Suppose
$k_1$ and $k_2$ are both roots of the dispersion relation for
linearized water waves of infinite depth,
\begin{equation}\label{resonance_deep}
  c^2 = gk^{-1}+\tau k.
\end{equation}
Here the gravitational acceleration $g$, wave speed $c$, and
  surface tension $\tau$ are held constant when solving for the wave
  numbers $k_1$ and $k_2$. If $k_2=Kk_1$ with $K\ge2$ an integer, the
$K^\text{th}$ harmonic will enter a modified Stokes expansion for traveling
waves of wavelength $2\pi/k_1$ at order $\veps^{\text{max(K-2,1)}}$
instead of $\veps^K$, where $\veps$ is the expansion coefficient of
the fundamental mode. This resonance occurs because the two waves
travel at the same speed under the linearized water wave
equations. Bridges and Dias \cite{bridges1996spatially} consider a
generalization in which the wave numbers $k_1$ and $k_2$ are
irrationally related.  They use a spatial Hamiltonian structure to
construct weakly nonlinear approximations of spatially quasi-periodic
traveling gravity-capillary waves for two special cases: deep water
and shallow water. This inspired us to develop a conformal mapping
framework for computing spatially quasi-periodic, fully nonlinear
traveling gravity-capillary waves, which is the topic of the companion
paper \cite{quasi:trav}.

We show in \cite{quasi:trav} that these spatially quasi-periodic
  traveling waves come in two-parameter families in which the
  amplitudes of the base modes with wave numbers $k_1$ and $k_2$ serve
  as bifurcation parameters. The wave speed and surface tension depend
  nonlinearly on these parameters as well.  Akers \emph{et al.}
  \cite{akers2019periodic} have proved existence of similar
  two-parameter families of traveling waves for the case of a
  two-fluid hydro-elastic interface. They develop an integral equation
  formulation of the equations governing traveling hydro-elastic waves
  such that the linearization about any state is a compact
  perturbation of the identity and use global bifurcation theory to
  establish existence and uniqueness results. They show that the
  nullspace of the linearized operator about the flat rest state has
  dimension one or two, and is two if and only if the
  non-dimensionalized wave numbers $k_1<k_2$ that travel with a given
  speed are integers. When they are integers, Akers \emph{et al.}
  distinguish resonant and non-resonant cases depending on whether
  $k_2/k_1$ is an integer. The non-resonant case leads to a smooth
  two-parameter family of traveling waves with wave speed and surface
  tension depending nonlinearly on the amplitude parameters. This is
  the case most analogous to the spatially quasi-periodic traveling
  waves that we compute in \cite{quasi:trav}.

The present paper focuses on the more general spatially quasi-periodic
initial value problem, which we use to validate the traveling wave
computations of \cite{quasi:trav} and explore new dynamic phenomena.
In recent years, conformal mapping methods have proved useful
  for studying two-dimensional traveling
  \cite{choi1999exact, milewski:10, wang2013two, dyachenko2016branch}
  and time-dependent \cite{dyachenko1996analytical,
    dyachenko1996nonlinear, dyachenko2001dynamics, zakharov2002new,
    li2004numerical, milewski:10, dyachenko:newell:16,
    turner:bridges} water waves with periodic boundary conditions.
We introduce a Hilbert transform for
quasi-periodic functions to compute the normal velocity and maintain a
conformal parametrization of the free surface.  This leads to a
numerical method to compute the time evolution of solutions of the
Euler equations from arbitrary quasi-periodic initial data.
Following the definitions in
  \cite{moser1966theory,dynnikov2005topology}, we represent a general
  quasi-periodic function $u(\alpha)$ in one dimension by a periodic
  function $\tilde u(\bds\alpha)$ on a $d$-dimensional torus,
  i.e.~$u(\alpha)=\tilde u(\bds k\alpha)$ for $\alpha\in\mbb R$, where
  $\bds k=(k_1,\dots,k_d)$. The $k_i$ are assumed to be linearly
  independent over the integers. We take these basic wave numbers
  $k_i$ and the initial conditions on the torus as given, focusing on
  the $d=2$ case.  This leaves open the important question of how best
  to measure a one-dimensional wave profile or velocity potential and
  identify its quasi-periods and corresponding torus function.

We present two variants of the numerical method, one in a high-order
explicit Runge-Kutta framework and one in an exponential
time-differencing (ETD) framework. The former is suitable for the case
of zero or small surface tension while the latter makes use of the
small-scale decomposition \cite{HLS94,HLS01} to eliminate stiffness
due to surface tension. The conformal mapping method has not
  been implemented in an ETD framework before, even for periodic
  boundary conditions.  We present a
convergence study of the methods as well as a large-scale computation
of a quasi-periodic wave in which some of the wave crests overturn
when evolved forward in time while others do not. Due to the
  torus representation of solutions, there are infinitely many wave
  crests and no two of them evolve in exactly the same way. The
computation involves over 33 million degrees of freedom evolved over
5400 time steps to maintain double-precision accuracy.

We include four appendices that cover various technical aspects
of this work. In Appendix~\ref{sec:one:one}, we prove a theorem
establishing sufficient conditions for an analytic function $z(w)$ to
map the lower half-plane topologically onto a semi-infinite region
bounded above by a parametrized curve and for $1/|z_w|$ to be
uniformly bounded. In Appendix~\ref{sec:families}, we study families
of quasi-periodic solutions obtained by introducing phases in the
reconstruction formula for extracting 1D quasi-periodic functions from
periodic functions on a torus. This enables us to prove that if all
the solutions in the family are single-valued and have no vertical
tangent lines, the solutions are also quasi-periodic in the original
graph-based formulation of the Euler equations. We also present a
simple procedure for computing the change of variables from the
conformal representation to the graph representation.  This appears
  to be a new result even for periodic boundary conditions. In
  Appendix~\ref{sec:etd} we provide details on how to implement the
  equations of motion in an exponential time-differencing framework to
  avoid stepsize limitations due to stiffness caused by surface
  tension.  And in Appendix~\ref{sec:3d}, we discuss the equations of
  motion for spatially quasi-periodic water waves in three dimensions
  and outline possible alternatives to the conformal mapping approach.


\section{Mathematical Formulation} \label{mathematical_formulation}


In this section, we review the governing equations for
gravity-capillary waves in both physical space and
  conformal space.  We then extend the
conformal mapping framework to allow for spatially quasi-periodic
solutions. For simplicity, we initially assume the wave profile
$\eta(x,t)$ remains single-valued. This assumption is relaxed when
discussing the conformal formulation, and an example of a wave in
which some of the peaks overturn as time advances is presented in
Section~\ref{sec:overturn}.


\subsection{Governing Equations in Physical Space}
\label{sec:gov:eqs}

Gravity-capillary waves of infinite depth are governed by the
two-dimensional free-surface Euler equations \cite{zakharov1968stability,CraigSulem}
\begin{equation} \label{general_initial}
  \eta(x, 0) = \eta_0(x), \qquad \varphi(x, 0) = \varphi_0(x), \qquad t=0, \quad
  x\in \mathbb{R},
\end{equation}
\begin{equation}\label{harmonic_function}
  \begin{aligned}
    \Phi_{xx} + \Phi_{yy} &= 0, \qquad -\infty < y < \eta(x, t),\\
    \Phi_y &\to 0, \qquad y \to -\infty,\\
    \Phi &= \varphi, \qquad y = \eta(x, t),
  \end{aligned}
\end{equation}
\begin{equation}\label{eta_phi_evolve}
  \eta_t = \Phi_y - \eta_x\Phi_x, \qquad y = \eta(x, t),
\end{equation}
\begin{equation}\label{Bernoulli}
  \varphi_t = \Phi_y\eta_t-\frac{1}{2}\Phi_x^2-\frac{1}{2}\Phi_y^2-
  g\eta+\tau\frac{\eta_{xx}}{(1+\eta_x^2)^{3/2}}+C(t), \qquad y = \eta(x, t),
\end{equation}
where $x$ is the horizontal coordinate, $y$ is the vertical
coordinate, $t$ is the time, $\Phi(x, y, t)$ is the velocity
potential in the fluid, $\eta(x, t)$ is the free surface
elevation,
\begin{equation}\label{eq:varphi:def}
  \varphi(x, t) = \Phi(x, \eta(x, t), t)
\end{equation}
is the boundary value of the velocity potential on the free surface,
$g$ is the vertical acceleration due to gravity and $\tau$ is the
coefficient of surface tension. Following
  \cite{zakharov1968stability, CraigSulem}, only the surface variables
  $\eta$ and $\varphi$ are evolved in time; the velocity potential
  $\Phi$ in the bulk fluid is reconstructed from $\eta$ and $\varphi$
  by solving (\ref{harmonic_function}), which causes the problem to be
  nonlocal.  The function $C(t)$ in the Bernoulli condition
(\ref{Bernoulli}) is an arbitrary integration constant that is allowed
to depend on time but not space. When the domain is periodic or
  quasi-periodic, one can choose $C(t)$ so that the mean value
of $\varphi(x,t)$ remains constant in time, where the mean is
  defined as $\lim_{a\rightarrow\infty}\frac1{2a}\int_{-a}^a \varphi(x,t)\,dx$.


\subsection{The Quasi-Periodic Hilbert Transform} \label{sec:quasi}


We find that a conformal mapping representation of the free surface
greatly simplifies the solution of the Laplace equation for the
velocity potential in the quasi-periodic setting. In this section, we
establish the properties of the Hilbert transform that will be needed
to study quasi-periodic water waves in a conformal mapping framework.

As defined in \cite{moser1966theory,dynnikov2005topology},
a quasi-periodic, real analytic function $u(\alpha)$ is a function
of the form
\begin{equation}\label{general_quasi_form}
  u(\alpha) = \tilde u(\bds{k} \alpha), \qquad
  \tilde u(\bds\alpha) = \sum_{\boldsymbol{j}\in\mathbb{Z}^d}\hat{u}_{\boldsymbol{j}}
  e^{i\la\boldsymbol{j},\,\boldsymbol{\alpha}\ra}, \qquad
  \alpha\in\mbb R, \;\; \bds\alpha,\bds k \in \mathbb{R}^d,
\end{equation}
where $\la\cdot , \,\cdot\ra$ denotes the standard inner
product in $\mathbb{R}^d$ and $\tilde u$ is a periodic, real analytic
function defined on the $d$-dimensional torus
\begin{equation}
  \mathbb{T}^d := \mbb R^d\big/(2\pi\mbb Z)^d. 
\end{equation}
Entries of the vector $\boldsymbol{k}$ are called the basic wave
  numbers (or basic frequencies) of $u$ and are required to be
  linearly independent over $\mathbb{Z}$. If $\bds{k}$ is
  given, one can reconstruct the Fourier coefficients
  $\hat u_{\bds j}$ from $u$ via
\begin{equation}\label{uhat:from:u}
  \hat{u}_{\bds{j}} = \lim_{a\to\infty} \frac{1}{2a} 
  \int_{-a}^au(\alpha) e^{-i\la \bds{j}, \bds{k}\ra\alpha} d\alpha,
  \qquad \bds{j}\in\mathbb{Z}^d.
\end{equation}
A similar averaging formula holds for functions in the more general
class of almost periodic functions \cite{moser1966theory,bohr:book,
  fink:book, amerio:book, hino:almost:periodic}, which is the closure
with respect to uniform convergence on $\mbb R$ of the set of
trigonometric polynomials $p(x)=\sum_{n=1}^N c_n e^{i\kappa_nx}$.
Before taking limits to obtain the closure, this set includes
polynomials of any degree and there is no restriction on the real
numbers~$\kappa_n$.  Within the framework of almost periodic
functions, one obtains quasi-periodic functions if one assumes the
$\kappa_n$ in the approximating polynomials are integer linear
combinations of a fixed, finite set of basic wave numbers
$k_1,\dots,k_d$.

We have not attempted to formulate the water wave problem in the
  full generality of almost periodic functions, and instead assume the
  basic wave numbers are given and the torus representation
  (\ref{general_quasi_form}) is available. Thus, the average over
  $\mbb R$ on the right-hand side of (\ref{uhat:from:u}) can be
  replaced by the simpler Fourier coefficient formula
  \begin{equation}\label{uhat:from:util}
    \hat u_{\bds j} = \frac1{(2\pi)^d}\int_{\mbb T^d}
    \tilde u(\bds\alpha)e^{-i\la\bds
      j,\bds\alpha\ra} \,d\alpha_1\cdots\, d\alpha_d.
  \end{equation}
  Our assumption that $\tilde u(\bds\alpha)$ is real analytic is
  equivalent to the conditions that $\hat u_{-\bds j}=\overline{\hat
    u_{\bds j}}$ for $\bds j\in\mbb Z^d$ and there exist positive
  numbers $M$ and $\sigma$ such that $|\hat u_{\bds j}|\le
  Me^{-\sigma\|\bds j\|}$, i.e.~the Fourier modes $\hat{u}_{\bds{j}}$
  decay exponentially as $\|\bds j\|\rightarrow\infty$. This
  is proved e.g.~in Lemma 5.6 of \cite{broer:book}.

Next we define the projection operators $\Projection$ and
$\Projection_0$ that act on $u$ and $\tilde u$ via
\begin{equation}\label{eq:proj}
  \Projection = \operatorname{id} - \Projection_0, \qquad
  \Projection_0 [u] = \Projection_0 [\tilde u] = \hat{u}_{\boldsymbol{0}}
  = \frac{1}{(2\pi)^d} \int_{\mathbb{T}^d}
  \tilde u(\boldsymbol{\alpha})\, d\alpha_1\cdots\, d\alpha_d.
\end{equation}
Note that $\Projection$ projects onto the space of zero-mean functions
while $\Projection_0$ returns the mean value, viewed as a constant
function on $\mathbb R$ or $\mathbb T^d$. There
  are two versions of $P$ and $P_0$, one acting on quasi-periodic
  functions defined on $\mbb R$ and one acting on torus functions
  defined on $\mbb T^d$.

Given $u(\alpha)$ as in (\ref{general_quasi_form}), the most
  general bounded analytic function $f(w)$ in the lower half-plane
  whose real part agrees with $u$ on the real axis has the form
\begin{equation}\label{eq:f:from:u}
  f(w) = \hat u_{\bds 0} + i\hat v_{\bds 0} + \sum_{\la\bds j,\bds k\ra<0}
  2 \hat u_{\bds j}e^{i\la\bds j,\bds k\ra w}, \qquad (w=\alpha+i\beta\,,\; \beta\le0)
\end{equation}
where $\hat v_{\bds 0}\in\mbb R$ and the sum is over all $\bds
j\in\mbb Z^d$ satisfying $\la\bds j,\bds k\ra<0$. The imaginary part
of $f(z)$ on the real axis is given by
\begin{equation}\label{eq:imf:v}
  v(\alpha) = \tilde v(\bds k\alpha), \qquad \tilde v(\bds\alpha) =
  \sum_{\bds j\in\mathbb{Z}^d} \hat{v}_{\bds j}e^{i\la\bds j,\bds
    \alpha\ra}, \qquad \hat{v}_{\bds j} = i\sgn(\la\bds j,\bds
    k\ra)\hat u_{\bds j}, \quad (\bds j\ne0)
\end{equation}
where $\sgn(q)\in\{1,0,-1\}$ depending on whether $q>0$, $q=0$ or
$q<0$, respectively. Similarly, given $v(\alpha)$ and requiring
$(\im f)\vert_\mbb R=v$ yields (\ref{eq:f:from:u}) with $\hat u_{\bds
  j}$ replaced by $i\hat v_{\bds j}$. We introduce a quasi-periodic
Hilbert transform to compute $v$ from $u$ or
$u$ from $v$,
\begin{equation}\label{eq:HuHv}
  v = \hat v_{\bds 0} - H[u], \qquad\quad
  u = \hat u_{\bds 0} + H[v],
\end{equation}
where the constant $\hat v_{\bds 0}=P_0[v]$ or $\hat u_{\bds
  0}=P_0[u]$ is a free parameter when computing $v$ or $u$,
respectively. $\Hilbert$ returns the ``zero-mean'' solution,
i.e.~$\Projection_0\Hilbert[u]=0$.  Uniqueness of the bounded
extension from $u$ or $v$ to $f$ up to the additive constant $i\hat
v_{\bds 0}$ or $\hat u_{\bds 0}$ follows from two well-known results:
the only bounded solution of the Laplace equation on a half-space
satisfying homogeneous Dirichlet boundary conditions is identically
zero \cite{axler:harmonic}, and the harmonic conjugate of the zero
function on a connected domain is constant.

\begin{definition}
  The Hilbert transform of a quasi-periodic, analytic function
  $u(\alpha)$ of the form (\ref{general_quasi_form}) is defined to be
  \begin{equation}\label{eq:hilb:def}
    \Hilbert[u](\alpha) = \sum\limits_{\boldsymbol{j}\in\mathbb{Z}^d}
            (-i)\sgn(\la\boldsymbol{j},\,\boldsymbol{k}\ra)
            \hat{u}_{\boldsymbol{j}}
            e^{i \la\boldsymbol{j},\,\boldsymbol{k}\ra\alpha}.
  \end{equation}
\end{definition}

  This agrees with the standard definition \cite{dyachenko2016branch}
  of the Hilbert transform as a Cauchy principal value integral:
  \begin{equation}\label{eq:Hf:PV}
    \Hilbert[u](\alpha) = \frac{1}{\pi}\PV
    \int_{-\infty}^\infty\frac{u(\xi)}{\alpha-\xi}\,d\xi.
  \end{equation}  
  Indeed, it is easy to show that for functions of the form $u(\alpha)
  = e^{i\rho \alpha}$ with $\rho$ real, the integral in
  (\ref{eq:Hf:PV}) gives $\Hilbert[u](\alpha) =
  -i\sgn(\rho)e^{i\rho\alpha}$.
  For extensions to the \emph{upper}
  half-plane, the sum in (\ref{eq:f:from:u}) is over $\la\bds j,\bds
  k\ra>0$, the last formula in (\ref{eq:imf:v}) becomes $\hat{v}_{\bds
    j} = -i\sgn(\la\bds j,\bds k\ra)\hat u_{\bds j}$, and the signs in
  front of $\Hilbert[u]$ and $\Hilbert[v]$ in (\ref{eq:HuHv}) are
  reversed.

\begin{remark}\label{rmk:f:quasi}
  As with $P$ and $P_0$, there is an analogous operator on $L^2(\mbb
    T^d)$ such that $\Hilbert[u](\alpha)=H[\tilde u](\bds k\alpha)$.
  The formula is
  \begin{equation}
    \Hilbert[\tilde u](\bds\alpha) = \sum\limits_{\boldsymbol{j}\in\mathbb{Z}^d}
            (-i)\sgn(\la\boldsymbol{j},\,\boldsymbol{k}\ra)
            \hat{u}_{\boldsymbol{j}}
            e^{i \la\boldsymbol{j},\,\boldsymbol{\alpha}\ra}.
  \end{equation}
  If necessary for clarity, one can also write $\Hilbert_{\bds
    k}[\tilde u]$ to emphasize the dependence of $\Hilbert$ on $\bds
  k$.  $\Hilbert$ commutes with the shift operator
  $S_{\bds\theta}[\tilde u](\bds\alpha) = \tilde
  u(\bds\alpha+\bds\theta)$, so if $\tilde v=\hat v_{\bds 0} -
  \Hilbert[\tilde u]$ and $\hat u_{\bds 0}=P_0[\tilde u]$, then
  $v(\alpha;\bds\theta)=\tilde v(\bds k\alpha+\bds\theta)$ is related
  to $u(\alpha;\bds\theta)=\tilde u(\bds k\alpha+\bds\theta)$ by
  (\ref{eq:HuHv}). Also, if $f(z)$ in (\ref{eq:f:from:u}) is the
  bounded analytic extension of $(u+iv)(\alpha)=\tilde u(\bds
      k\alpha) +i\tilde v(\bds k\alpha)$ to the lower half-plane, we
  have
  \begin{equation}\label{eq:f:tilde}
    f(w) = \tilde f(\bds k\alpha,\beta), \qquad
    (w=\alpha+i\beta\,,\;\beta\le0),
  \end{equation}
  where $\tilde f(\bds\alpha,\beta) = \hat u_{\bds0} + i\hat v_{\bds0} +
  \sum_{\la\bds j,\bds k\ra<0} 2[\hat u_{\bds j}e^{-\la\bds j,\bds k\ra\beta}]
  e^{i\la\bds j,\bds\alpha\ra}$ is periodic in $\bds\alpha$
  for fixed $\beta\le0$.  The bounded analytic extension of
  $[u(\alpha;\bds\theta)+iv(\alpha;\bds\theta)]$ to the lower
  half-plane is then given by $f(\alpha+i\beta;\bds\theta)=
  \tilde f(\bds k\alpha+\bds\theta,\beta)$.
\end{remark}


\subsection{The Conformal Mapping} \label{the_conformal_mapping}


We consider a time dependent conformal mapping that maps the conformal
domain
\begin{equation}
  \mbb C^- = \{\alpha+i\beta\,:\, \alpha\in\mbb R,\,\beta<0\}
\end{equation}
to the fluid domain
\begin{equation}
  \Omega(t) = \{x + iy\,:\, x\in\mbb R,\,y<\eta(x,t)\}.
\end{equation}
This conformal mapping, denoted by $z(w, t)$, is assumed to
extend continuously to $\overline{\mbb C^-}$ and maps the real line
$\beta = 0$ to the free surface
\begin{equation}
  \Gamma(t) = \{x + iy\,:\, y = \eta (x, t)\}.
\end{equation}
We express $z(w, t)$ as
\begin{equation} \label{conformal_mapping}
  z(w, t) = x(w, t) + i y(w, t), \qquad\quad (w=\alpha+i\beta).
\end{equation}
We also introduce the notation $\zeta=z\vert_{\beta=0}$,
$\xi=x\vert_{\beta=0}$ and $\eta=y\vert_{\beta=0}$ so that the free
surface is parametrized by
\begin{equation}\label{zeta:xi:eta}
  \zeta (\alpha, t) = \xi(\alpha, t) + i \eta(\alpha, t), \qquad\quad
  (\alpha\in\mathbb R).
\end{equation}
This allows us to denote a generic field point in the physical fluid
by $z=x+iy$ while simultaneously discussing points $\zeta=\xi+i\eta$
on the free surface. To avoid ambiguity, we will henceforth denote
the free surface elevation function from the previous section by
$\eta^\ph(x, t)$. Thus,
\begin{equation}\label{eta_chain_rule}
  \eta(\alpha,t) = \eta^\ph(\xi(\alpha,t),t), \qquad
  \eta_\alpha = \eta^\ph_x\xi_\alpha, \qquad \eta_t =
  \eta^\ph_x\xi_t + \eta^\ph_t.
\end{equation}
The parametrization (\ref{zeta:xi:eta}) is more general than
(\ref{eta_chain_rule}) in that it allows for overturning waves. In
deriving the equations of motion for $\zeta(\alpha,t)$ and
$\varphi(\alpha,t)$ in Section~\ref{sec:gov:conf} below, we will
indicate the modifications necessary to handle the case of overturning
waves. In particular, as discussed in Appendix~\ref{sec:one:one},
$\Gamma(t)$ is defined in this case as the image of $\zeta(\cdot,t)$,
which is assumed to be injective on $\mbb R$, and $\Omega(t)$ can be
obtained from $\Gamma(t)$ using the Jordan curve theorem.

The conformal map is required to remain a bounded distance
from the identity map in the lower half-plane. Specifically, we
require that
\begin{equation}\label{y_boundary}
  |z(w,t)-w|\le M(t) \quad\qquad (w=\alpha+i\beta\,,\; \beta\le0),
\end{equation}
where $M(t)$ is a uniform bound that could vary in time.  The
Cauchy integral formula implies that $|z_w-1|\le M(t)/|\beta|$, so
at any fixed time,
\begin{equation}\label{eq:zw:lim}
  z_w\rightarrow1 \quad \text{ as } \quad \beta\to-\infty.
\end{equation}
Our goal is to investigate the case when the free surface is
quasi-periodic in $\alpha$. This differs from conformal mappings
discussed in \cite{meiron1981applications, dyachenko1996analytical,
  dyachenko2001dynamics,
  zakharov2002new, li2004numerical, milewski:10}, where it is assumed
to be periodic.

In the present work, $\eta$ is assumed to have two spatial
quasi-periods, i.e.~at any time it has the form
(\ref{general_quasi_form}) with $d=2$ and $\bds{k}=[k_1,k_2]^T$. Since
$k_1$ and $k_2$ are irrationally related, we assume without loss of
generality that $k_1=1$ and $k_2 = k$, where $k$ is irrational:
\begin{equation}\label{eq:eta:tilde}
    \eta(\alpha,t) = \tilde\eta(\alpha,k\alpha,t), \qquad
    \tilde\eta(\alpha_1,\alpha_2,t) = \sum_{j_1, j_2\in\mbb Z}
  \hat{\eta}_{j_1, j_2}(t)e^{i(j_1\alpha_1+j_2\alpha_2)}.
\end{equation}
Here $\hat{\eta}_{-j_1, -j_2}(t) = \overline{\hat{\eta}_{j_1,j_2}(t)}$
since $\tilde\eta(\alpha_1,\alpha_2,t)$ is real-valued. Since
$w\mapsto[z(w,t)-w]$ is bounded and analytic on $\mbb C^-$ and its imaginary
part agrees with $\eta$ on the real axis, there is a real number $x_0$
(possibly depending on time) such that
\begin{equation}\label{eq:xi:from:eta}
  \xi(\alpha, t) = \alpha + x_0(t) + \Hilbert[\eta](\alpha, t).
\end{equation}
Using (\ref{eq:zw:lim}) and\, $\im [z_w\vert_{\beta=0}]=\eta_\alpha$\, or
differentiating (\ref{eq:xi:from:eta}) gives
\begin{equation}\label{hilbert_xi}
  \xi_\alpha(\alpha, t) = 1 + \Hilbert[\eta_\alpha](\alpha, t).
\end{equation}
We use a tilde to denote the periodic functions on the torus
that correspond to the quasi-periodic parts of $\xi$, $\zeta$ and
$z$,
\begin{equation}\label{eq:xi:zeta}
  \begin{gathered}
    \xi(\alpha,t) = \alpha + \tilde\xi(\alpha,k\alpha,t), \qquad
    \zeta(\alpha,t) = \alpha + \tilde\zeta(\alpha,k\alpha,t), \\
    z(\alpha+i\beta,t)=\left(\alpha+i\beta \right)+
    \tilde z(\alpha,k\alpha,\beta,t),
    \qquad (\beta\le 0).
  \end{gathered}
\end{equation}
Specifically, $\tilde\xi = x_0(t) + \Hilbert[\tilde\eta]$,
$\tilde\zeta= \tilde\xi + i\tilde\eta$, and
\begin{equation}\label{eq:z:tilde}
  \tilde z(\alpha_1,\alpha_2,\beta,t) =
  x_0(t)+i\hat\eta_{0,0}(t)+\sum_{j_1+j_2k<0}
  \left(2i\hat\eta_{j_1,j_2}(t)e^{-(j_1+j_2k)\beta}\right)
  e^{i(j_1\alpha_1+j_2\alpha_2)}.
\end{equation}
While the mean surface height remains constant in physical space,
$\hat\eta_{0,0}(t)$ generally varies in time. Since the modes
$\hat\eta_{j_1,j_2}$ are assumed to decay exponentially, there is a
uniform bound $M(t)$ such that $|\tilde
z(\alpha_1,\alpha_2,\beta,t)|\le M(t)$ for $(\alpha_1,\alpha_2)\in\mbb
T^2$ and $\beta\le 0$.  In Appendix~\ref{sec:one:one} we show that as
long as the free surface $\zeta(\alpha,t)$ does not self-intersect at
a given time $t$, the mapping $w\mapsto z(w,t)$ is an analytic
isomorphism of the lower half-plane onto the fluid region.

\subsection{The Complex Velocity Potential}
\label{sec:cx:vel}

Let $\Phi^\ph(x,y,t)$ denote the velocity potential in physical space
from Section~\ref{sec:gov:eqs} above and let
$W^\ph(x+iy,t)= \Phi^\ph(x,y,t)+i\Psi^\ph(x,y,t)$ be the
complex velocity potential, where $\Psi^\ph$ is the stream function.
Using the conformal mapping (\ref{conformal_mapping}), we pull
back these functions to the lower half-plane and define
\begin{equation*}
  W(w,t) = \Phi(\alpha,\beta,t)+i\Psi(\alpha,\beta,t) = W^\ph(z(w,t),t), \qquad\quad
  (w=\alpha+i\beta).
\end{equation*}
We also define $\varphi=\Phi\vert_{\beta=0}$ and
$\psi=\Psi\vert_{\beta=0}$ and use (\ref{harmonic_function}) and
(\ref{eta_chain_rule}) to obtain
\begin{equation}\label{eq:phi:tilde}
  \varphi(\alpha,t) = \varphi^\ph(\xi(\alpha,t),t), \qquad
  \psi(\alpha,t) = \psi^\ph(\xi(\alpha,t),t),
\end{equation}
where $\psi^\ph(x,t)=\Psi^\ph(x,\eta^\ph(x,t),t)$.  We assume
$\varphi$ is quasi-periodic with the same quasi-periods as $\eta$,
\begin{equation*}
  \varphi(\alpha, t) = \tilde\varphi(\alpha,k\alpha,t), \qquad
  \tilde\varphi(\alpha_1,\alpha_2,t) =
  \sum_{j_1, j_2\in\mathbb{Z}}
  \hat{\varphi}_{j_1, j_2}(t) e^{i(j_1\alpha_1+j_2\alpha_2)}.
\end{equation*}
The fluid velocity $\nabla\Phi^\ph(x,y,t)$ is assumed to decay to zero
as $y\rightarrow-\infty$ (since we work in the lab frame).  From
(\ref{eq:zw:lim}) and the chain rule (see (\ref{chain_rule}) below),
$dW/dw\rightarrow0$ as
$\beta\to-\infty$. Thus, $\psi_\alpha = -H[\varphi_\alpha]$. Writing
this as $\pa_\alpha[\psi+H\varphi]=0$, we conclude that
\begin{equation}\label{hilbert_phi}
  \psi(\alpha, t) = -\Hilbert[\varphi](\alpha, t).
\end{equation}
Here we have set the integration constant to zero and assumed
$\Projection_0[\varphi] =
\hat{\varphi}_{0,0}(t) = 0$ and $\Projection_0[\psi] =
\hat{\psi}_{0,0}(t) = 0$, which is allowed
since $\Phi$ and $\Psi$ can be modified by additive constants (or
  functions of time only) without affecting the fluid motion.


\subsection{Governing Equations in Conformal Space}
\label{sec:gov:conf}


Following
  \cite{dyachenko1996analytical,choi1999exact,zakharov2002new,
    li2004numerical,viotti2014conformal,turner:bridges}, we present a
derivation of the equations of motion for surface water waves in a
conformal mapping formulation, modified as needed to handle
quasi-periodic solutions. We also justify the assumption that
  $z_t/z_\alpha$ remains bounded in the lower half-plane, which we
  have not seen discussed previously in the literature.

From the chain rule,
\begin{equation} \label{chain_rule}
  \frac{dW}{dw} = \frac{dW^\ph}{dz}\cdot\frac{dz}{dw} \qquad \Rightarrow \qquad
  \Phi^\ph_x + i \Psi^\ph_x = \frac{\Phi_\alpha + i\Psi_\alpha}{x_\alpha+iy_\alpha}.
\end{equation}
Evaluating (\ref{chain_rule}) on the free surface gives
\begin{equation} \label{chain_rule_surface}
  \Phi^\ph_x = \frac{\varphi_\alpha \xi_\alpha + \psi_\alpha \eta_\alpha}J, \qquad
  \Phi^\ph_y = -\Psi^\ph_x = \frac{\varphi_\alpha \eta_\alpha - \psi_\alpha \xi_\alpha}J,
  \qquad  J = \xi_\alpha^2 + \eta_\alpha^2.
\end{equation}
Using (\ref{eta_chain_rule}) and (\ref{chain_rule_surface})
in (\ref{eta_phi_evolve}) and multiplying by $\xi_\alpha$, we obtain
\begin{equation}\label{kinematic_conformal_1}
  \eta_t \xi_\alpha- \xi_t\eta_\alpha = -\psi_\alpha.
\end{equation}
This states that the normal velocity of the free surface is equal to
the normal velocity of the fluid,
$\bds{\hat n}\cdot(\xi_t,\eta_t) = \bds{\hat{n}}\cdot\nabla\Phi^\ph$,
where $\bds{\hat n} = (-\eta_\alpha,\xi_\alpha)/\sqrt{J}$. This can
also be obtained by tracking a fluid particle $x_p(t)+iy_p(t)=
\zeta(\alpha_p(t),t)$ on the free surface. We have $\dot
x_p=\xi_\alpha\dot\alpha_p+\xi_t=\Phi^\ph_x$ and $\dot
y_p=\eta_\alpha\dot\alpha_p+\eta_t=\Phi^\ph_y$, which leads to
(\ref{kinematic_conformal_1}) after eliminating $\dot\alpha_p$.
This argument does not assume the free surface is a graph,
i.e.~(\ref{kinematic_conformal_1}) is also valid for overturning
waves.

Next we define a new function,
\begin{equation}
  q:= \frac{\zeta_t}{\zeta_\alpha} = \frac{(\xi_t\xi_\alpha+\eta_t\eta_\alpha) +
    i (\eta_t\xi_\alpha-\xi_t\eta_\alpha)}J = \frac{
  (\xi_t\xi_\alpha+\eta_t\eta_\alpha) - i \psi_\alpha}J.
\end{equation}
Since $q$ is quasi-periodic in $\alpha$ and extends analytically to
the lower half-plane via $z_t/z_\alpha$, the real and imaginary part of
$q$ can be related by the Hilbert transform. Here we have assumed that
$z_t/z_\alpha$ is bounded, which will be justified below. Thus,
\begin{equation}\label{kinematic_conformal_2}
  \frac{\xi_t\xi_\alpha+\eta_t\eta_\alpha}J =
  -\Hilbert\left[\frac{\psi_\alpha}{J}\right] + C_1,
\end{equation}
where $C_1$ is an arbitrary integration constant that may depend on
time but not space. Let $\bds{\hat s} =
(\xi_\alpha,\eta_\alpha)/\sqrt{J}$ denote the unit tangent vector to
the curve. Equation (\ref{kinematic_conformal_2}) prescribes the
tangential velocity\, $\bds{\hat s}\cdot (\xi_t,\eta_t)$\, of points
on the curve in terms of the normal velocity in order to maintain a
conformal parametrization. Note that the tangent velocity of the curve
differs from that of the underlying fluid particles. This is similar
in spirit to a method of Hou, Lowengrub and Shelley
\cite{HLS94,HLS01}, who proposed a tangential velocity that maintains
a uniform parametrization of the curve (rather than a conformal one);
see also \cite{choi1999exact, turner:bridges, da:jm:jw:laypot,
  quasi:trav}.  Combining (\ref{kinematic_conformal_1}) and
(\ref{kinematic_conformal_2}), we obtain the kinematic boundary
conditions in conformal space,
\begin{equation} \label{kinematic_conformal_3}
  \begin{pmatrix}
    \xi_t \\
    \eta_t
  \end{pmatrix}
  = 
  \begin{pmatrix}
    \xi_\alpha & -\eta_\alpha\\
    \eta_\alpha & \xi_\alpha
  \end{pmatrix}
  \begin{pmatrix}
    -\Hilbert\left[\frac{\psi_\alpha}{J}\right] + C_1\\
    -\frac{\psi_\alpha}{J}
  \end{pmatrix}.
\end{equation}
The right-hand side can be interpreted as complex multiplication of
$z_\alpha$ with $z_t/z_\alpha$. Since both functions are analytic in
the lower half-plane, their product is, too. Thus, $\xi_t$ is related
to $\eta_t$ via the Hilbert transform (up to a constant).
The constant is determined by comparing (\ref{eq:xi:from:eta}) with
(\ref{kinematic_conformal_3}), which gives
\begin{equation}\label{eq:x0:evol}
  \frac{dx_0}{dt} = \Projection_0\left[
    \xi_\alpha\left(-\Hilbert\left[\frac{\psi_\alpha}J\right]
      + C_1\right) + \frac{\eta_\alpha\psi_\alpha}J\right].
\end{equation}
The three most natural choices of $C_1$ are
\begin{equation}\label{eq:C1:opt}
  \begin{alignedat}{2}
    (a) \;\; C_1&=0: & &
    \text{evolve $x_0(t)$ via (\ref{eq:x0:evol})}, \\
    (b) \;\; C_1&=\Projection_0\big[\xi_\alpha\Hilbert[\psi_\alpha/J] -
      \eta_\alpha\psi_\alpha/J\big]:& \qquad
    &x_0(t)=0, \\
    (c) \;\; C_1&=\big[\Hilbert[\psi_\alpha/J] -
      \eta_\alpha\psi_\alpha/(\xi_\alpha J)\big]_{\alpha=0}:& \qquad
    &\xi(0,t) = 0.
  \end{alignedat}
\end{equation}
In options ($b$) and ($c$), the evolution equation ensures that
$dx_0/dt=0$ and $\xi_t(0,t)=0$, respectively; we have assumed the
initial conditions satisfy $x_0(0)=0$ or $\xi(0,0)=0$,
respectively. Option ($c$) amounts to setting
$x_0(t)=-\Hilbert[\eta](0,t)$ in (\ref{eq:xi:from:eta}). This arguably
leads to the most natural parametrization, but would have a problem if
the vertical part of an overturning wave crosses $\alpha=0$. Indeed,
such a crossing would lead to $\xi_\alpha(0,t)=0$ at some time $t$ in
the denominator of (\ref{eq:C1:opt}c). We recommend option ($b$)
in this scenario.

\begin{remark}\label{rmk:ab:same}
  In infinite depth as considered here, if $\tilde f$ and $\tilde g$
  are torus functions on $\mbb T^d$ and $f(\alpha)=\tilde f(\bds
    k\alpha)$, $g(\alpha)=\tilde g(\bds k\alpha)$, the identity
  \begin{equation}\label{eq:H:ident}
    P_0\big[ fg - (Hf)(Hg)\big] = \sum_{\la\bds j,\bds k\ra=0} \hat
    f_{\bds j}\hat g_{-\bds j}
  \end{equation}
  is easily proved, where the sum is over $\bds j\in\mbb Z^d$
  satisfying $\la\bds j,\bds k\ra=0$. In the periodic case with $d=1$
  and $\bds k=(1)$ or the quasi-periodic case with $\bds
  k=(k_1,\dots,k_d)$ and the $k_i$ linearly independent over
  the integers, the
  right-hand side of (\ref{eq:H:ident}) is $\hat f_{\bds 0} \hat
  g_{\bds 0}$. This simplifies (\ref{eq:C1:opt}b) to $C_1=0$ and
  (\ref{eq:x0:evol}) to $dx_0/dt=C_1$, i.e.~cases (a) and (b) in
  (\ref{eq:C1:opt}) coincide.  However, this only works in infinite
  depth as the finite-depth Hilbert transform with symbol
  $-i\tanh\big(\la\bds j,\bds k\ra h\big)$ does not satisfy
  (\ref{eq:H:ident}). Here $h$ is the fluid depth in conformal space,
  which evolves in time to maintain constant fluid depth in physical
  space \cite{li2004numerical,turner:bridges}. Finite-depth
  quasi-periodic water waves
  will be investigated in future work.
\end{remark}

Next we evaluate the Bernoulli equation\, $\Phi^\ph_t +
\frac12\big|\nabla\Phi^\ph\big|^2 + \frac p\rho + gy = C_2$\, at the
free surface to obtain an evolution equation for $\varphi(\alpha,t)$.
Here $C_2$ is an arbitrary integration constant that may depend on
time but not space.  The pressure at the free surface is determined by
the Laplace-Young condition, $p=p_0-\rho\tau\kappa$, where $\kappa$ is
the curvature, $\rho\tau$ is the surface tension, and $p_0$ is a
constant that can be absorbed into $C_2$ and set to zero.  From
(\ref{chain_rule}) or (\ref{chain_rule_surface}), we know
$\big|\nabla\Phi^\ph\big|^2=(\varphi_\alpha^2+\psi_\alpha^2)/J$ on the
free surface. Finally, differentiating
$\varphi(\alpha,t)=\Phi^\ph(\xi(\alpha,t),\eta(\alpha,t),t)$ and using
(\ref{chain_rule_surface}) and (\ref{kinematic_conformal_3}), we
obtain
\begin{equation}\label{eq:bernoulli:conf}
  \varphi_t = \underbrace{\big(\Phi^\ph_x\,,\,\Phi^\ph_y\big)\begin{pmatrix}
    \xi_\alpha & -\eta_\alpha\\
    \eta_\alpha & \xi_\alpha
  \end{pmatrix}}_{\jd\big(\varphi_\alpha\,,\,-\psi_\alpha\big)}
  \begin{pmatrix}
    -\Hilbert\left[\psi_\alpha/J\right] + C_1\\
    -\psi_\alpha/J
  \end{pmatrix}
  - \frac{\varphi_\alpha^2+\psi_\alpha^2}{2J} - g\eta + \tau\kappa + C_2.
\end{equation}
We choose $C_2$ so that $\Projection_0[\varphi_t]=0$.  In conclusion,
we obtain the following governing equations for spatially
quasi-periodic gravity-capillary waves in conformal space
\begin{equation} \label{general_conformal}
  \begin{gathered}
    \xi_\alpha = 1 + \Hilbert[\eta_\alpha], \qquad
    \psi = -\Hilbert[\varphi], \qquad
    J = \xi_\alpha^2 + \eta_\alpha^2, \qquad
    \chi = \frac{\psi_\alpha}{J}, \\[2pt]
    \text{
        choose $C_1$, e.g.~as in (\ref{eq:C1:opt}), \qquad
        compute\, $\jd\frac{dx_0}{dt}$\, in
      (\ref{eq:x0:evol}) if necessary,} \\[2pt]
    \eta_t = -\eta_\alpha \Hilbert[\chi] - \xi_\alpha\chi +
    C_1\eta_\alpha, \qquad
    \kappa = \frac{\xi_\alpha\eta_{\alpha\alpha} -
      \eta_\alpha\xi_{\alpha\alpha}}{J^{3/2}}, \\[-4pt]
    \varphi_t = P\bigg[\frac{\psi_\alpha^2 - \varphi_\alpha^2}{2J} -
    \varphi_\alpha\Hilbert[\chi] + C_1\varphi_\alpha - g\eta +
    \tau\kappa\bigg].
  \end{gathered}
\end{equation}
Note that these equations govern the evolution of $x_0$, $\eta$ and
$\varphi$, which determine the state of the system. The functions
$\xi$, $\psi$, $J$, $\chi$ and $\kappa$ are determined at any moment
by $x_0$, $\eta$ and $\varphi$ through the auxiliary equations in
(\ref{general_conformal}).  We emphasize that $C_1$ can be chosen
arbitrarily as long as $dx_0/dt$ satisfies (\ref{eq:x0:evol}). The
special cases (\ref{eq:C1:opt}b) and (\ref{eq:C1:opt}c) lead to nice
formulas for $x_0(t)$ without having to evolve (\ref{eq:x0:evol})
numerically.  An alternative approach was proposed by
  Li \emph{et al.} \cite{li2004numerical}, who set $C_1=0$
  (by not introducing it) and avoid writing down a differential
  equation for $x_0(t)$ by instead solving both the $\xi_t$ and
  $\eta_t$ equations in (\ref{kinematic_conformal_3}).

In deriving (\ref{general_conformal}) from
(\ref{kinematic_conformal_1}) and (\ref{eq:bernoulli:conf}), we had to
assume $z_t/z_\alpha$ remains bounded in the lower
half-plane. Conditions that ensure the boundedness of $1/z_\alpha$ are
given in Appendix~\ref{sec:one:one}. We note that $z_t/z_\alpha$ is
automatically bounded in the converse direction, where
(\ref{kinematic_conformal_1}) and (\ref{eq:bernoulli:conf}) are
derived from (\ref{general_conformal}). In more detail, when solving
(\ref{general_conformal}), $z_t/z_\alpha$ is constructed first,
before~$z_t$, as the bounded extension of the quasi-periodic function
with imaginary part $(-\psi_\alpha/J)$ to the lower
half-plane. Equation (\ref{kinematic_conformal_3}) then defines $z_t$
as the product of this function by $z_\alpha$, which is also bounded
since $\xi_\alpha=1+\Hilbert[\eta_\alpha]$.  Thus, the first component
of each side of (\ref{kinematic_conformal_3}) is related to the
corresponding second component by the Hilbert transform, up to a
constant. Since the second components are equal (i.e.~the $\eta_t$
  equation holds), the $\xi_t$ equation also holds --- the constants
are accounted for by (\ref{eq:x0:evol}).  Left-multiplying
(\ref{kinematic_conformal_3}) by the row vector
$[-\eta_\alpha,\xi_\alpha]$ gives the kinematic condition
(\ref{kinematic_conformal_1}), as required.

Equations (\ref{general_conformal}) break down if $J$ becomes zero
somewhere on the curve. Such a singularity would arise, for example,
if the wave profile were to form a corner in finite time. To our
knowledge, it remains an open question whether the free-surface Euler
equations can form such a corner.

Often we wish to verify that a given curve
$\zeta(\alpha,t)=\xi(\alpha,t)+i\eta(\alpha,t)$ and velocity potential
$\varphi(\alpha,t)$ satisfy the conformal version of the water wave
equations. We say that $(\zeta,\varphi)$ satisfy
(\ref{general_conformal}) if $\xi$ and $\eta$ remain conformally
related via (\ref{eq:xi:from:eta}), which determines $x_0(t)$, and if
$x_0$, $\eta$ and $\varphi$ satisfy (\ref{general_conformal}) with
$C_1(t)$ obtained from (\ref{eq:x0:evol}) using $P_0[\xi_\alpha]=1$.
As noted above, these equations imply the kinematic condition
(\ref{kinematic_conformal_1}) and Bernoulli equation
(\ref{eq:bernoulli:conf}). If necessary, one should replace the given
$\varphi(\alpha,t)$ by $\Projection[\varphi(\cdot,t)]$ before checking
that (\ref{general_conformal}) is satisfied.

\begin{remark}\label{rmk:torus}
  Equations (\ref{general_conformal}) can be interpreted as an
  evolution equation for the functions
  $\tilde\zeta(\alpha_1,\alpha_2,t)$ and
  $\tilde\varphi(\alpha_1,\alpha_2,t)$ on the torus $\mbb T^2$.  The
  $\alpha$-derivatives are replaced by the directional derivatives
  $(\pa_{\alpha_1}+k\pa_{\alpha_2})$, which we still denote by a
    subscript $\alpha$, e.g.~$\tilde\eta_\alpha =
    (\pa_{\alpha_1}+k\pa_{\alpha_2})\tilde\eta$, and, as noted in
  Remark~\ref{rmk:f:quasi} above, the Hilbert transform becomes a
  two-dimensional Fourier multiplier operator with symbol
  $(-i)\sgn(j_1+j_2k)$. The pseudo-spectral method we propose in
  Section~\ref{sec:num} below is based on this
  representation. Equation (\ref{eq:C1:opt}c) becomes
  \begin{equation}\label{eq:C1:opt2}
    C_1 = \left[\Hilbert\left[\frac{\tilde\psi_\alpha}{\tilde J}\right] -
      \frac{\tilde\eta_\alpha\tilde\psi_\alpha}{
        (1+\tilde\xi_\alpha)\tilde J}
      \right]_{(\alpha_1,\alpha_2)=(0,0)}, \qquad\quad
    \tilde\xi(0,0,t)=0,
  \end{equation}
  where $\tilde J=(1+\tilde\xi_\alpha)^2 + \tilde\eta_\alpha^2$.
  Note that $\xi_\alpha$ in (\ref{general_conformal}) is replaced
  by
  \begin{equation}
    \widetilde{\xi_\alpha} = 1 + \tilde\xi_\alpha,
  \end{equation}
  which is the one place this notation becomes awkward.  Using
  (\ref{eq:xi:from:eta}) and (\ref{eq:xi:zeta}), $\tilde\zeta$ is
  completely determined by $x_0(t)$ and $\tilde\eta$, so only these
  have to be evolved --- the formula for $\tilde\xi_t$ in
  (\ref{kinematic_conformal_3}) is redundant as long as
  (\ref{eq:x0:evol}) is satisfied. If both components of
  $\tilde\zeta(\alpha_1,\alpha_2,t)$ are given, we say that
  $(\tilde\zeta,\tilde\varphi)$ satisfy the torus version of
  (\ref{general_conformal}) if there is a continuously differentiable
  function $x_0(t)$ such that $\tilde\xi=x_0(t)+\Hilbert[\tilde\eta]$
  and if $x_0$, $\tilde\eta$ and $\tilde\varphi$ satisfy the torus
  version of (\ref{general_conformal}) with $C_1=dx_0/dt +
  P_0\big[(1+\tilde\xi_\alpha)H[\tilde\psi_\alpha/\tilde J] -
    \tilde\eta_\alpha\tilde\psi_\alpha/\tilde J\big]$.  As noted
    in Remark~\ref{rmk:ab:same}, if $k$ is irrational, one can
    verify that $P_0\left[\tilde{\xi}_\alpha
      H[\tilde{\psi}_\alpha/\tilde{J}] -
      \tilde{\eta}_\alpha\tilde{\psi}_\alpha/\tilde{J}\right] = 0$,
    hence $d x_0/dt = C_1$.
\end{remark}

We show in Appendix~\ref{sec:families} that solving the torus version
of (\ref{general_conformal}) yields a three-parameter family of
one-dimensional solutions of the form
\begin{equation}\label{eq:family:full}
  \begin{aligned}
    \zeta(\alpha,t\,;\,\theta_1,\theta_2,\delta) &= \alpha +
    \delta + \tilde\zeta(
      \theta_1+\alpha,\theta_2+k\alpha,t), \\
    \varphi(\alpha,t\,;\,\theta_1,\theta_2) &= \tilde\varphi(
      \theta_1+\alpha,\theta_2+k\alpha,t),
  \end{aligned} \qquad
  \left( \begin{aligned}
      \alpha\in\mbb R, \, t\ge0 \\
      \theta_1,\theta_2,\delta\in\mbb R
    \end{aligned}\right).
\end{equation}
We also show that if all the waves in this family are single-valued
and have no vertical tangent lines, there is a corresponding family of
solutions of the Euler equations in the original graph-based
formulation of (\ref{general_initial})--(\ref{Bernoulli}) that are
quasi-periodic in physical space. A precise statement is given in
Theorem~\ref{thm:quasi:phys1} and the discussion that follows.


\section{Numerical Method}
\label{sec:num}


In this section, we describe a pseudo-spectral time-stepping strategy
for evolving water waves with spatially quasi-periodic initial
conditions.  The evolution equations (\ref{general_conformal}) for
$\eta$ and $\varphi$ are nonlinear and involve computing derivatives,
antiderivatives and Hilbert transforms of quasi-periodic functions.
Let $f$ denote one of these functions (e.g.~$\eta$, $\varphi$ or
  $\chi$) and let $\tilde f$ denote the corresponding periodic
function on the torus,
%
\begin{equation}\label{quasi_f_form}
  f(\alpha) = \tilde f(\alpha,k\alpha), \qquad
  \tilde f(\alpha_1,\alpha_2) = 
  \sum\limits_{j_1, j_2 \in \mathbb{Z}}
  \hat{f}_{j_1, j_2} e^{i(j_1\alpha_1+j_2k\alpha_2)}, \qquad
  (\alpha_1, \alpha_2)\in\mathbb{T}^2.
\end{equation}
The functions $f_\alpha$ and $\Hilbert[f]$ then correspond to
\begin{equation} \label{quasi_deriv_hilbert}
  \begin{gathered}
    \wtil{f_\alpha} (\alpha_1, \alpha_2) = \sum\limits_{j_1, j_2\in \mathbb{Z}}
    i(j_1 + j_2 k)
    \hat{f}_{j_1, j_2} e^{i(j_1\alpha_1+j_2\alpha_2)}, \\
    \wtil{\Hilbert[f]} (\alpha_1, \alpha_2) = \sum\limits_{j_1, j_2\in \mathbb{Z}}
            (-i)\text{sgn}(j_1 + j_2 k) \hat{f}_{j_1, j_2}
            e^{i(j_1\alpha_1+j_2\alpha_2)}.
  \end{gathered}
\end{equation}
We propose a pseudo-spectral method in which each such $f$ that arises
in the formulas (\ref{general_conformal}) is represented by the values
of $\tilde f$ at $M_1\times M_2$ equidistant gridpoints on the torus
$\mbb T^2$,
\begin{equation}\label{eq:T2:discr}
  \tilde f_{m_1,m_2} = \tilde f(2\pi m_1/M_1\,,\,2\pi m_2/M_2), \qquad\quad
  (0\le m_1<M_1\,,\,0\le m_2<M_2).
\end{equation}
We visualize a $90^\circ$ rotation between the matrix $\tilde F$
holding the entries $\tilde f_{m_1,m_2}$ and the collocation points in
the torus.  The columns of the matrix correspond to horizontal slices
of gridpoints while the rows of the matrix correspond to vertical
slices indexed from bottom to top.  The nonlinear operations in
(\ref{general_conformal}) consist of products, powers and division;
they are carried out pointwise on the grid. Derivatives and the
Hilbert transform are computed in Fourier space via
(\ref{quasi_deriv_hilbert}).
To plot the solution, we also need to compute an antiderivative to
get $\xi$ from $f=\xi_\alpha$. This involves dividing $\hat
f_{j_1,j_2}$ by $i(j_1 + j_2 k)$ when $(j_1,j_2)\ne(0,0)$ and
adjusting the $(0,0)$ mode to obtain $\xi(0,t)=0$.

Since the functions $f$ that arise in the computation are real-valued,
we use the real-to-complex (`r2c') version of the two-dimensional
discrete Fourier transform.  The `r2c' transform of a one-dimensional
array of length $M$ (assumed even) is given by
\begin{equation}\label{eq:r2c}
  \{g_m\}_{m=0}^{M-1} \quad \mapsto \quad
  \{\hat g_j\}_{j=0}^{M/2} \quad , \quad
  \hat g_j = \frac1M\sum_{m=0}^{M-1} g_m e^{-2\pi ijm/M}.
\end{equation}
In practice, the $\hat g_j$ are computed simultaneously in $O(M\log
  M)$ time rather than by this formula.  The fully complex (`c2c')
transform of this (real) data would give additional values $\hat g_j$
with $M/2+1\le j\le M-1$. These extra entries are actually aliased
values of negative-index modes; they are redundant due to $\hat
g_j=\hat g_{j-M} = \overline{\hat g_{M-j}}$.  Since the imaginary
components of $\hat g_0$ and $\hat g_{M/2}$ are zero, the number of
real degrees of freedom on both sides of (\ref{eq:r2c}) is $M$. The
Nyquist mode $\hat g_{M/2}$ requires special attention.  Setting $\hat
g_{M/2}=1$ and the other modes to zero yields $g_m=\cos(\pi m)
=(-1)^m$. The derivative and Hilbert transform of this mode are taken
to be zero since they would involve evaluating $\sin(M\alpha/2)$ at
the gridpoints $\alpha_m=2\pi m/M$.

The two-dimensional `r2c' transform can be computed by applying
one-dimensional `r2c' transforms in the $x$-direction (i.e.~to the
  columns of $\tilde F$) followed by one-dimensional `c2c' transforms
in the $y$-direction (i.e.~to the rows of $\tilde F$):
\begin{equation}\label{eq:r2c:fhat}
  \hat f_{j_1,j_2} = \frac1{M_2}\sum_{m_2=0}^{M_2-1}\left(\frac1{M_1}\sum_{m_1=0}^{M_1-1}
    \tilde f_{m_1,m_2} e^{-2\pi ij_1m_1/M_1}\right)e^{-2\pi ij_2m_2/M_2}, \quad
  \left(\begin{gathered}
      0\le j_1\le M_1/2 \\
      -M_2/2< j_2\le M_2/2
    \end{gathered}\right).
\end{equation}
The `r2c' routine in the FFTW library actually returns the index
range $0\le j_2<M_2$, but we use $\hat f_{j_1,j_2-M_2}=\hat
f_{j_1,j_2}$ to de-alias the Fourier modes and map the indices
$j_2>M_2/2$ to their correct negative values. The missing entries with
$-M_1/2<j_1<0$ are determined implicitly by
\begin{equation}\label{eq:f:j1j2:sym}
  \hat f_{-j_1,-j_2}=\overline{\hat f_{j_1,j_2}}.
\end{equation}
This imposes additional constrains on the computed Fourier modes,
namely
\begin{equation}\label{eq:r2c:constraints}
  \begin{gathered}
    \im\{\hat f_{0,0}\}=0,\quad\im\{\hat f_{M_1/2,0}\}=0,\quad
    \im\{\hat f_{0,M_2/2}\}=0,\quad \im\{\hat f_{M_1/2,M_2/2}\}=0, \\
    \hat f_{0,-j_2}=\overline{\hat f_{0,j_2}}, \qquad
    \hat f_{M_1/2,-j_2}=\overline{\hat f_{M_1/2,j_2}}, \qquad
    (1\le j_2\le M_2/2-1),
  \end{gathered}
\end{equation}
where we also used $\hat f_{-M_1/2,j_2}=\hat f_{M_1/2,j_2}$.  This
reduces the number of real degrees of freedom in the complex
$(M_1/2+1)\times M_2$ array of Fourier modes to $M_1M_2$. When
computing $f_\alpha$ and $\Hilbert[f]$ via
(\ref{quasi_deriv_hilbert}), the Nyquist modes with $j_1=M_1/2$ or
$j_2=M_2/2$ are set to zero.  Otherwise the formulas
(\ref{quasi_deriv_hilbert}) respect the constraints
(\ref{eq:r2c:constraints}) and the `c2r' transform
reconstructs real-valued functions $\wtil{f_\alpha}$ and
$\wtil{\Hilbert[f]}$ from their Fourier modes.

The evolution equations (\ref{general_conformal}) are not stiff when
the surface tension parameter is small or vanishes, but become
moderately stiff for larger values of $\tau$. We find that the 5th and
8th order explicit Runge-Kutta methods of Dormand and Prince
\cite{hairer:I} work well for smaller values of $\tau$, and
exponential time-differencing (ETD) methods
\cite{cox:matthews,kassam,berland,whalen:etd,chen:wilkening} work well
generally. This will be demonstrated in
  Sections~\ref{sec:trav:rslts} and~\ref{sec:overturn} below.  In the
ETD framework, we follow the basic idea of the small-scale
decomposition for removing stiffness from interfacial flows
\cite{HLS94,HLS01} and write the
evolution equations (\ref{general_conformal})
in the form
\begin{equation}\label{eq:ssd:L}
  \begin{pmatrix}
    \eta_t \\ \varphi_t
  \end{pmatrix} = L\begin{pmatrix}
    \eta \\ \varphi
  \end{pmatrix} + \mc{N}, \qquad
  L = \begin{pmatrix}
    0 & \Hilbert\partial_\alpha \\
    -(g\Projection-\tau\partial_{\alpha\alpha}) & 0
    \end{pmatrix},
\end{equation}
where $\Projection$ is the projection in (\ref{eq:proj}),
$\Hilbert$ is the Hilbert transform in (\ref{eq:hilb:def}), and
\begin{equation}\label{eq:ssd:N}
  \mc N = \begin{pmatrix}
    -\eta_\alpha \Hilbert[\chi] - \big(\xi_\alpha\chi - \psi_\alpha\big)
    + C_1\eta_\alpha \\
    \Projection\left[ \frac{\psi_\alpha^2-\varphi_\alpha^2}{2J} -
      \varphi_\alpha \Hilbert[\chi] + C_1\varphi_\alpha + \tau(\kappa -
        \eta_{\alpha\alpha})\right]\end{pmatrix}.
\end{equation}
Note that $\mc N$ is obtained by subtracting the terms included in $L$
from (\ref{general_conformal}).  In particular, $\psi_\alpha$ in
  (\ref{eq:ssd:N}) is $-\Hilbert\pa_\alpha\varphi$ from (\ref{eq:ssd:L}).
The eigenvalues of\, $L$\,
are\, $\pm i\sqrt{|j_1+j_2k|\big(g+\tau(j_1+j_2k)^2\big)}$, so the leading
source of stiffness is dispersive. This $3/2$ power growth rate of the
eigenvalues of the leading dispersive term with respect to wave number
is typical of interfacial fluid flows with surface tension
\cite{HLS94,HLS01}. For stiffer problems such as the
Benjamin-Ono and KdV equations, the growth rate is
faster (quadratic and cubic, respectively) and it becomes essential to
use a semi-implicit or exponential time-differencing scheme to avoid
severe time-stepping restrictions.  Here it is less critical, but still
useful. Further details on how to implement (\ref{eq:ssd:L}) and
(\ref{eq:ssd:N}) in the ETD framework are given in
Appendix~\ref{sec:etd}.

In both the explicit Runge-Kutta and ETD methods, as explained above,
the functions evolved in time are $\tilde\eta(\alpha_1,\alpha_2,t)$
and $\tilde\varphi(\alpha_1,\alpha_2,t)$, sampled on the uniform
$M_1\times M_2$ grid covering $\mbb T^2$.  At the end of each time
step, we apply a 36th order filter \cite{hou:li:07,HLS94} with
Fourier multiplier
\begin{equation}\label{eq:filter}
  \rho(j_1,j_2) = \left\{\begin{array}{cc}
  0 & \hspace*{-3pt} j_1=M_1/2 \text{\, or\, } |j_2|=M_2/2, \\[3pt]
    \exp\big(\!-36\big[(2j_1/M_1)^{36}+
        (2j_2/M_2)^{36}\big]\big) & \text{otherwise.}
  \end{array}\right.
\end{equation}
In all the computations reported below, we used the same number of
gridpoints in the $\alpha_1$ and $\alpha_2$-directions,
$M_1=M_2=M$.  It is easy to check \emph{a-posteriori} that the
  Fourier modes decay sufficiently (e.g.~to machine precision) by the
  time the filter deviates appreciably from 1. If they do not, the
  calculation can be repeated with a larger value of $M$. This will be
  demonstrated in Section~\ref{sec:overturn} below.


\section{Numerical Results} \label{numerical_results}


In this section, we compute spatially quasi-periodic solutions of the
initial value problem (\ref{general_conformal}) with $k = 1/\sqrt{2}$
and $g$ normalized to 1. First we validate the traveling wave
computations of \cite{quasi:trav} and compare the accuracy and
efficiency of the ERK and ETD schemes in a convergence study. We then
consider more complex dynamics in which some of the wave peaks overturn.

\subsection{Traveling waves}\label{sec:trav:rslts}

\begin{figure}
\includegraphics[width=\textwidth]{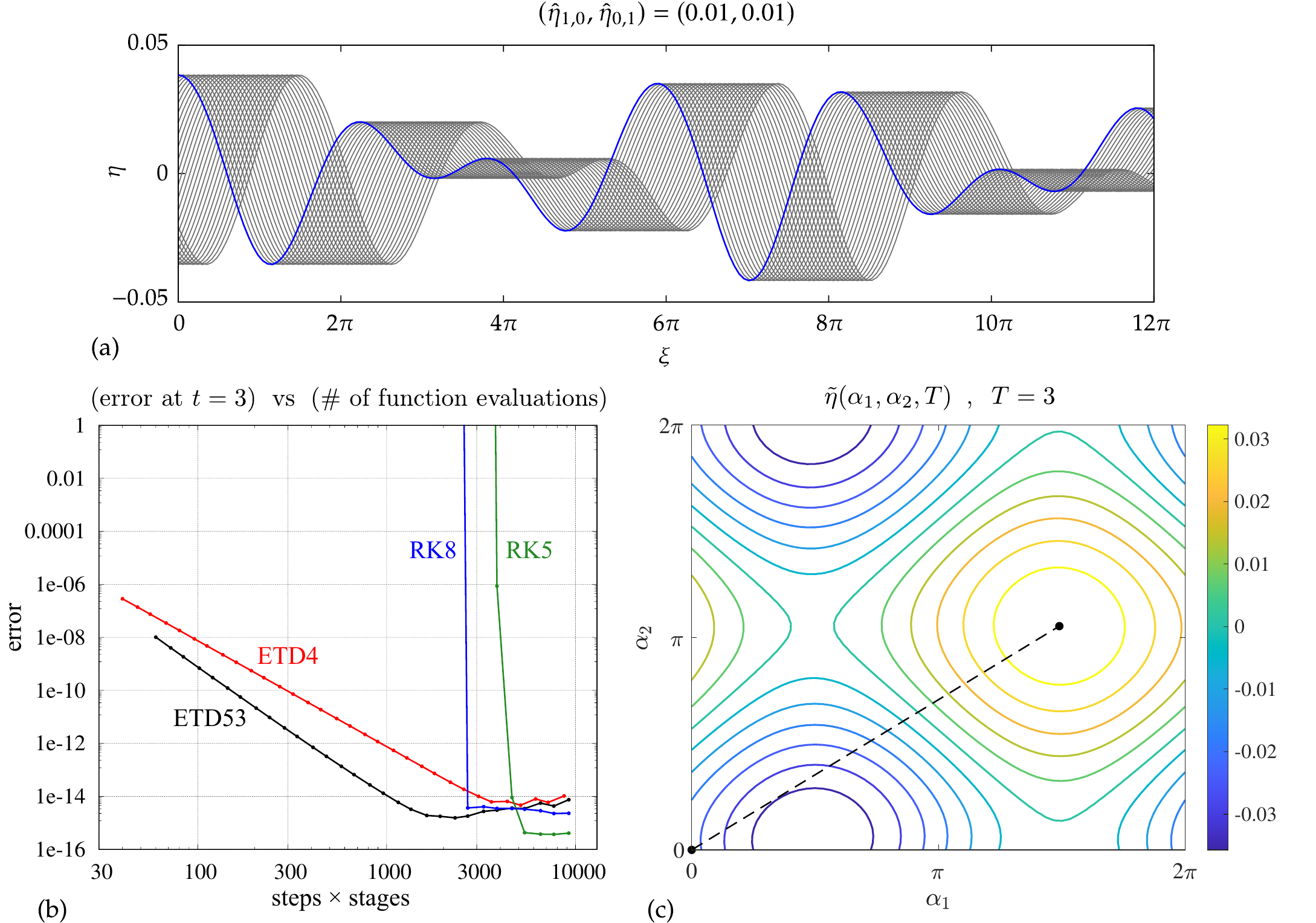}
\caption{\label{stepper_plot} Time evolution of a quasi-periodic
  traveling water wave computed in \cite{quasi:trav} and a convergence plot
  comparing the accuracy and efficiency of the proposed time-stepping
  schemes.}
\end{figure}

Figure~\ref{stepper_plot} shows the time evolution of a traveling wave
computed at $t=0$ by the algorithm of \cite{quasi:trav} with amplitude
parameters $\hat\eta_{1,0}=0.01$ and $\hat\eta_{0,1}=0.01$ and
  evolved to $t=T$ using the methods of Section~\ref{sec:num}, where
  $T=3$.  It also shows the error at the final time $T$ for various
choices of time-stepping scheme and number of time steps.  In
  all the computations of the figure, the torus functions
  $\tilde\eta(\alpha_1,\alpha_2,t)$ and
  $\tilde\varphi(\alpha_1,\alpha_2,t)$ are evolved on an $M\times M$
  mesh with $M=60$ gridpoints in each direction.  Panel (a) shows
snapshots of the solution in the lab frame at 30 equal time intervals
of size $T/30=0.1$.  Here we plotted every 30th step of the 5th
  order, 6 stage explicit Runge-Kutta method of Dormand and Prince
  \cite{hairer:I}, so the Runge-Kutta stepsize was $\Delta t=0.1/30$.
The initial condition is plotted with a thick blue line, and the wave
travels right at constant speed $c=1.552197$ in physical space.  The
solution is plotted over the representative interval $0\le x\le
12\pi$, though it extends in both directions to $\pm\infty$ without
exactly repeating.

Panel (b) shows the error in time-stepping this traveling wave
solution from $t=0$ to $t=3$ using the 5th and 8th order explicit
Runge-Kutta methods of Dormand and Prince \cite{hairer:I}, the 4th
order ETD scheme of Cox and Matthews \cite{cox:matthews, kassam}, and
the 5th order ETD scheme of Whalen, Brio and Moloney
\cite{whalen:etd}. These errors compare the numerical solution from
time-stepping the initial condition to the exact formula of how a
quasi-periodic traveling wave should evolve under
(\ref{general_conformal}) and (\ref{eq:C1:opt2}), which is worked out
in \cite{quasi:trav}.  If the initial wave profile has the torus
representation $\tilde\eta_0(\alpha_1,\alpha_2)$, we define
$\tilde\xi_0=\Hilbert[\tilde\eta_0]$ and
$\tilde\varphi_0=c\tilde\xi_0$.  By construction
  \cite{quasi:trav}, $\tilde\eta_0$ is an even function of
$\bds\alpha=(\alpha_1,\alpha_2)$, so $\tilde\xi_0$ and
$\tilde\varphi_0$ are odd. The exact traveling solution is then
\begin{equation}\label{eq:exact:a0}
  \begin{aligned}
    \tilde\eta_\text{exact}(\bds\alpha,t) &=
    \tilde\eta_0\big(\bds\alpha-\bds k\alpha_0(t)\big), \\
    \tilde\varphi_\text{exact}(\bds\alpha,t) &=
    \tilde\varphi_0\big(\bds\alpha-\bds k\alpha_0(t)\big),
  \end{aligned}
\end{equation}
where $\bds k=(1,k)$, $\alpha_0(t)=ct - \mc A(-\bds kct,0)$ and
$\mc A(\bds x,t)$ is a periodic function on $\mbb T^2$ defined
implicitly by (\ref{eq:mcA:def1}) below.  To derive
  (\ref{eq:exact:a0}), one makes use of the change of variables
  formula (\ref{eq:chg:vars:T2}) from physical space, where the wave
  speed is constant, to conformal space; see \cite{quasi:trav}. Note
that the waves in (\ref{eq:exact:a0}) do not change shape as they
  move through the torus in the direction $\bds k$, but the traveling
  speed $\alpha_0'(t)$ in conformal space varies in time in order to
  maintain $\tilde\xi(0,0,t)=0$ via (\ref{eq:C1:opt2}). The error
  plotted in panel (b) is the discrete norm at the final time
  computed, $T=3$:
\begin{equation*}
  \text{err} = \sqrt{\|\tilde\eta-\tilde\eta_{\text{exact}}\|^2 +
    \|\tilde\varphi-\tilde\varphi_{\text{exact}}\|^2}, \qquad
  \|\tilde\eta\|^2 = \frac1{M_1M_2}\sum_{m_1,m_2}
      \tilde\eta\left(\frac{2\pi m_1}{M_1},\frac{2\pi m_2}{M_2},T\right)^2.
\end{equation*}
The surface tension in this example ($\tau=1.410902$) is high enough
that once the stepsize is sufficiently small for the Runge-Kutta methods
to be stable, roundoff error dominates truncation error. So the errors
suddenly drop from very large values ($10^{20}$ or more) to machine
precision. By contrast, the error in the ETD methods decreases
steadily as the stepsize is reduced, indicating that the small-scale
decomposition introduced in (\ref{eq:ssd:L}) is successful in removing
stiffness from the equations of motion \cite{HLS94,HLS01}.

A contour plot of $\tilde\eta(\alpha_1,\alpha_2,T)$ at $t=T=3$ is
shown in panel (c) of Figure~\ref{stepper_plot}.  The dashed line
shows the trajectory from $t=0$ to $t=3$ of the wave crest that begins
at $(0,0)$ and continues along the path $\alpha_1=\alpha_0(t),$
$\alpha_2=k\alpha_0(t)$. We use Newton's method to solve the implicit
equation (\ref{eq:mcA:def1}) for $\mc A(\bds x,0)$ at each point of the
pseudo-spectral grid. We then use FFTW to compute the 2D Fourier
representation of $\mc A(\bds x,0)$, which can then be used to quickly
evaluate the function at any point.
We find that the Fourier modes of $\mc A(\bds x,0)$ decay to machine
precision on the $M\times M$ grid with $M=60$, corroborating the
assertion in Theorem~\ref{thm:quasi:phys1} below that $\mc A(x_1,x_2,t)$
is real analytic in $x_1$ and $x_2$.

\subsection{Overturning waves}\label{sec:overturn}
Next we present a spatially quasi-periodic water wave computation in
which some of the wave peaks overturn as they evolve while others do not.
Conformal mapping methods have been used previously to compute
  overturning waves. For example, Dyachenko and Newell
  \cite{dyachenko:newell:16} use this approach to study whitecapping in the
  ocean and Wang \textit{et al.}  \cite{wang2013two} use it
  to compute solitary and periodic overturning traveling flexural-gravity
  waves. The novelty of our work is the computation of a spatially
  quasi-periodic water wave in which every wave peak evolves
  differently, and only some of them overturn. Since torus functions
  are involved, the number of degrees of freedom is squared, leading
  to a large-scale computation.
For simplicity, we set the surface tension parameter, $\tau$, to
zero.

We first seek spatially periodic dynamics in which the initial
wave profile has a vertical tangent line that overturns when evolved
forward in time and flattens out when evolved backward in time. Through
trial and error, we selected the following parametric curves for the
initial wave profile and velocity potential of this auxiliary periodic
problem:
\begin{equation}\label{eq:ot:param}
  \xi_1(\bbb) = \bbb+\frac35\sin \bbb - \frac15\sin2\bbb, \qquad
  \begin{aligned}
    \eta_1(\bbb) &= -(1/2)\cos(\bbb+\pi/2.5), \\
    \varphi_1(\bbb) &= -(1/2)\cos(\bbb+\pi/4).
  \end{aligned}
\end{equation}
Note that $\xi_1'(\bbb)=0$ when $\bbb\in\pi+2\pi\mbb Z$, and otherwise
$\xi_1'(\bbb)>0$. Thus, vertical tangent lines occur where
$\xi_1(\bbb)\in\pi+2\pi\mbb Z$ and
$\eta_1(\bbb)=-0.5\cos(1.4\pi)=0.154508$; see
Figure~\ref{fig:periodic:ot}.

To convert (\ref{eq:ot:param}) to a conformal parametrization, we
search for $2\pi$-periodic functions $\eta_2(\alpha)$ and
$B_2(\alpha)$ and a number $x_2$ such that
\begin{equation}\label{eq:ot:conf}
  \alpha+x_2+H\left[\eta_2\right](\alpha) = \xi_1(\alpha+B_2(\alpha)), \qquad
  \eta_2(\alpha) = \eta_1(\alpha+B_2(\alpha)), \qquad B_2(0)=0.
\end{equation}
First we solve a simpler variant in which $x_2$ is absent and $B_2(0)$
is unspecified. Specifically, we solve
$\alpha+H\left[\eta_3\right](\alpha) = \xi_1(\alpha+B_3(\alpha))$,
$\eta_3(\alpha) = \eta_1(\alpha+B_3(\alpha))$ for $\eta_3(\alpha)$ and
$B_3(\alpha)$ on a uniform grid with $M=4096$ gridpoints on
$[0,2\pi)$ using Newton's method. The Hilbert transform is computed
with spectral accuracy in Fourier space. We then define $x_2$ as the
solution of $x_2+B_3(x_2)=0$ that is smallest in magnitude. We solve
this equation by a combination of root bracketing and Newton's method;
the result is $x_2=0.393458$. Finally, we define
$B_2(\alpha)=x_2+B_3(\alpha+x_2)$ and
$\eta_2(\alpha)=\eta_3(\alpha+x_2)$, which satisfy
(\ref{eq:ot:conf}).

\begin{figure}
  \begin{center}
    \includegraphics[width=.92\textwidth]{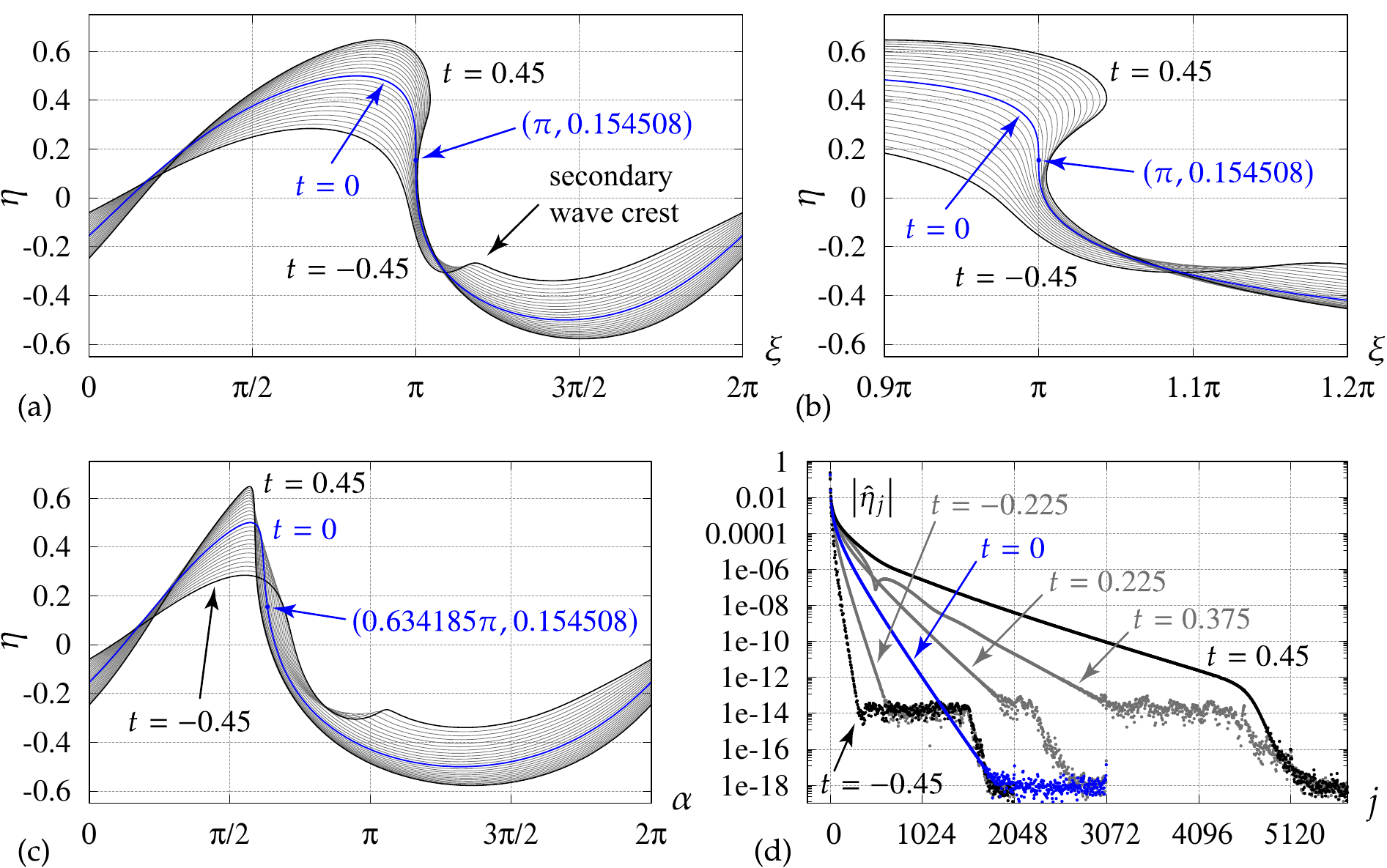}
  \end{center}
  \caption{\label{fig:periodic:ot} Time evolution of a spatially
    periodic water wave initialized via (\ref{eq:ot:param}) and
    evolved forward and backward in time to $t=\pm0.45$. Panels (a)
    and (b) show snapshots of the wave in physical space; panel (c)
    shows snapshots of $\eta(\alpha,t)$ in conformal space; and panel
    (d) shows snapshots of $|\hat\eta_j(t)|$ in Fourier space. The
    initial condition ($t=0$) is shown in blue in each plot.
  }
\end{figure}

As shown in Figure~\ref{fig:periodic:ot}, the initial conditions
$\eta_2(\alpha)$ and
\begin{equation}
  \varphi_2(\alpha)=\varphi_1(\alpha+B_2(\alpha))
\end{equation}
have the desired property that the wave overturns when evolved forward
in time and flattens out when evolved backward in time.  In other
  words, the wave becomes less steep in the neighborhood of the
  initial vertical tangent line when time is reversed. However, it
  does not evolve backward to a flat state. Instead, a secondary wave
  crest forms to the right of the initial wave crest and grows in
  amplitude as $t$ decreases. This secondary wave crest resembles the
  early stages of the fluid jets that were observed by Aurther
  \emph{et al.} \cite{shkoller} to form in the wave troughs when the
  initial condition $\eta_0(x)=\big(\frac13\sin x+\frac16\sin
    2x+\frac13\sin 3x)$ is evolved from rest in the graph-based
  formulation (\ref{general_initial}).

  The blue markers in panels (a) and (b) of
  Figure~\ref{fig:periodic:ot} show the location of the vertical
  tangent line in physical space at $t=0$.
The blue marker in panel (c) shows the corresponding point in
conformal space. When the wave overturns for $t>0$ in physical space,
it is because $\alpha\mapsto\xi(\alpha,t)$ no longer increases
monotonically. Indeed, we see in panel (c) that $\eta(\alpha,t)$
remains single-valued as a function of the conformal variable $\alpha$
but becomes very steep. This causes the Fourier mode amplitudes in
panel (d) to decay more slowly as $t$ increases. We used different
mesh sizes and timesteps in the regions $0\le t\le 0.3$, $0.3\le t\le
0.45$ and $0\ge t\ge-0.45$ to maintain spectral accuracy; details are
given below when discussing the quasi-periodic calculation. The
drop-off in $|\hat\eta_j|$ from $10^{-14}$ to $10^{-18}$ as $j$
approaches the Nyquist frequency $M/2$ is due to the 1D version of the
filter (\ref{eq:filter}), which is applied after each
timestep. Floating point errors of size $10^{-14}$ occur in the
discretization of the equations of motion while errors of size
$10^{-18}$ are due to having computed the inverse FFT of the filtered
data to get back to real space before taking the FFT again to
plot the Fourier data.

We turn the solution of this auxiliary periodic problem into a
spatially quasi-periodic solution by defining initial
conditions on the torus of the form
\begin{equation}\label{eq:ot:init:quasi}
  \tilde\eta_0(\alpha_1,\alpha_2) = \eta_2(\alpha_1), \qquad
  \tilde\varphi_0(\alpha_1,\alpha_2) = \varphi_2(\alpha_1)\cos(\alpha_2-q),
\end{equation}
where $q$ is a free parameter that we choose heuristically to be
$q=0.6k\pi=1.3329$ in order to make the first wave crest to the right
of the origin behave similarly to the periodic 1D solution of
Figure~\ref{fig:periodic:ot}. (This will be explained below).

\begin{figure} 
  \begin{center}
    \includegraphics[width=.92\textwidth]{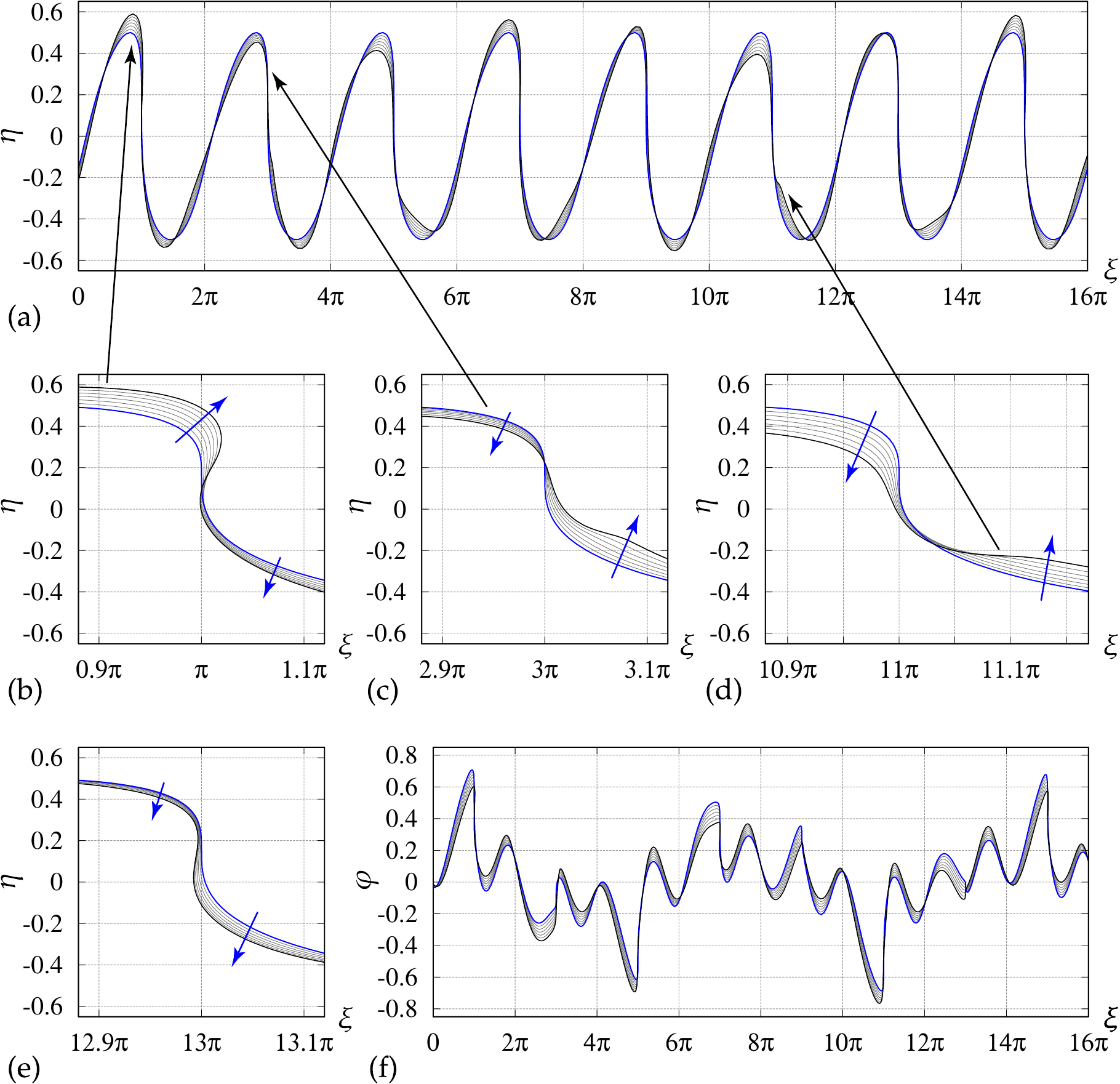}
  \end{center}
  \caption{\label{fig:ot:1} Snapshots in time of a spatially
    quasi-periodic water wave with a periodic initial wave profile
    with vertical tangent lines at $\xi=\pi+2\pi n$, $n\in\mbb Z$. A
    quasi-periodic initial velocity potential causes some of the 
    peaks to overturn for $t>0$ while others do not.
    Panels (a) and (f) show $\eta(\alpha,t)$ and $\varphi(\alpha,t)$ versus
    $\xi(\alpha,t)$ over $0\le x\le16\pi$ and $0\le t\le
    T=0.225$. Panels (b)--(e) show the results of panel (a) in more
    detail. The blue arrows show the direction of travel of the
    wave at various locations.}
\end{figure}

\begin{figure}
\includegraphics[width=\textwidth]{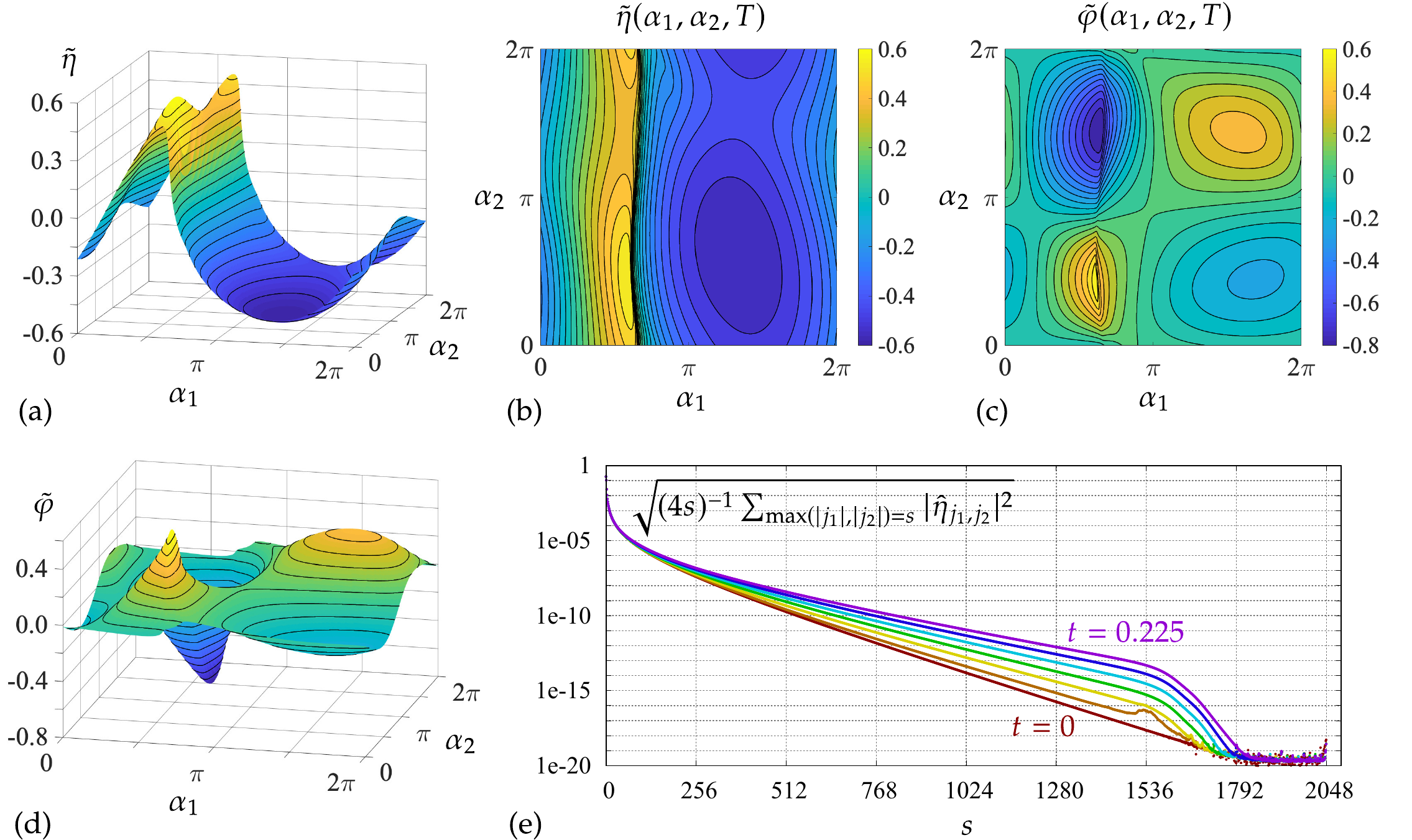}
\caption{\label{fig:ot:2} Surface and contour plots of the torus
  version of the solution plotted in Figure~\ref{fig:ot:1} at the
  final time $T=0.225$. The rapid dropoff in
  $\tilde\eta(\alpha_1,\alpha_2,t)$ over the window
  $0.6\pi\le\alpha_1\le0.667\pi$ persists from the initial state in
  which $\tilde\eta_0(\alpha_1,\alpha_2)$ does not depend on
  $\alpha_2$.  Panel (e) shows the exponential decay of Fourier modes
  with respect to the shell index $s$ at different times.}
\end{figure}

The results of the quasi-periodic calculation are summarized in
Figures~\ref{fig:ot:1} and~\ref{fig:ot:2}.  Panel (a) of
Figure~\ref{fig:ot:1} shows snapshots of the solution at $t=(\ell/6)T$
for $0\le\ell\le 6$ over the range $0\le\xi(\alpha)\le 16\pi$, where
$T=0.225$. The initial wave profile,
$\zeta_0(\alpha)=\xi_0(\alpha)+i\eta_0(\alpha)$ with
$\eta_0(\alpha)=\tilde\eta_0(\alpha,k\alpha)$, is plotted with a thick
blue line.
The wave profile is plotted with a thick black line at $t=T$ and with
thin grey lines at intermediate times.  Panel (b) zooms in on the
first wave in panel (a), which overturns as the wave crest moves up
and right while the wave trough moves down and left, as indicated by
the blue arrows. This is very similar (by design) to the forward
evolution of the auxiliary periodic wave of
Figure~\ref{fig:periodic:ot}, with initial conditions
$\eta_2(\alpha)$, $\varphi_2(\alpha)$. Panels (c) and (d) zoom in on
two other wave crests from panel (a) that flatten out (rather than
  overturn) as $t$ advances from 0 to $T$.  Panel (e) shows another
type of behavior in which the wave overturns due to the wave trough
moving down and left faster than the wave crest moves down and
left. Panel (f) shows the evolution of the velocity potential
$\varphi(\alpha,t)$ over $0\le t\le T$.  Unlike $\eta_0(\alpha)$, the
initial velocity potential
$\varphi_0(\alpha)=\tilde\varphi_0(\alpha,k\alpha)$ is not
$2\pi$-periodic due to the factor of $\cos(\alpha_2-q)$ in
(\ref{eq:ot:init:quasi}).

Panels (a) and (d) of Figure~\ref{fig:ot:2} show surface plots of
$\tilde\eta(\alpha_1,\alpha_2,T)$ and
$\tilde\varphi(\alpha_1,\alpha_2,T)$ at the final time computed,
$T=0.225$. The corresponding contour plots are shown in panels (b) and
(c). Initially, $\tilde\eta(\alpha_1,\alpha_2,0)$ depends only on
$\alpha_1$; however, by $t=T$, the dependence on $\alpha_2$ is clearly
visible. Although the waves overturn in some places when
$\eta(\alpha,t)= \tilde\eta(\alpha,k\alpha,t)$ is plotted
parametrically versus $\xi(\alpha,t)$ with $t>0$ held fixed, both
$\tilde\eta$ and $\tilde\varphi$ are single-valued functions of
$\alpha_1$ and $\alpha_2$ at all times.  Nevertheless, throughout the
evolution, $\tilde\eta(\alpha_1,\alpha_2,t)$ has a steep dropoff over
a narrow range of values of $\alpha_1$.  Initially,
$\tilde\eta_0(\alpha_1,\alpha_2)= \eta_2(\alpha_1) =
\eta_1(\alpha_1+B_2(\alpha_1))$ and the rapid dropoff occurs for
$\alpha_1$ near the solution of $\alpha_1+B_2(\alpha_1)=\pi$ (since
  the vertical tangent line occurs at $\xi_1(\bbb)+i\eta_1(\bbb)$ with
  $\bbb=\pi$). Using Newton's method, we find that this occurs at
$\alpha_1=0.634185\pi$. The blue curve in panel (c) of
Figure~\ref{fig:periodic:ot} gives $\eta_2(\alpha)$. If one
zooms in on this plot, one finds that $\eta_2(\alpha)$ decreases
rapidly by more than half its crest-to-trough height over the narrow
range $0.6\pi\le\alpha\le0.667\pi$. At later times,
$\tilde\eta(\alpha_1,\alpha_2,t)$ continues to drop off rapidly when
$\alpha_1$ traverses this narrow range in spite of the dependence on
$\alpha_2$. This can be seen in panel (b) of Figure~\ref{fig:ot:2},
where there is a high clustering of nearly vertical contour lines
separating the yellow-orange region from the blue region.  Over this
narrow window, $\tilde\varphi(\alpha_1,\alpha_2,t)$ also varies
rapidly with respect to $\alpha_1$.

Many gridpoints are needed to resolve these rapid variations with
spectral accuracy. Although $\xi_1(\bbb)$, $\eta_1(\bbb)$ and
$\varphi_1(\bbb)$ involve only a few nonzero Fourier modes, conformal
reparametrization via (\ref{eq:ot:conf}) vastly increases the Fourier
content of the initial condition. We used $M=6144$ gridpoints to
evolve the periodic auxiliary problem of Figure~\ref{fig:periodic:ot}
from $t=0$ to $t=0.3$  using the 8th order Runge-Kutta method of
  Dormand and Prince \cite{hairer:I} with stepsize $\Delta
  t=2.08333\times10^{-5}$. We then switched to $M=12288$ gridpoints
to evolve from $t=0.3$ to $t=0.45$ with $\Delta
  t=7.5\times10^{-6}$. In the reverse direction, we used $M=4096$
  gridpoints to evolve from $t=0$ to $t=-0.45$ with $\Delta
  t=-4.6875\times10^{-5}$.  Studying the Fourier modes in panel (d)
of Figure~\ref{fig:periodic:ot}, it appears that 4096 gridpoints (2048
  modes) are sufficient to maintain double-precision accuracy forward
or backward in time to $t=\pm0.225$. Using this as a guideline for the
quasi-periodic calculation, we evolved (\ref{general_conformal}) on a
$4096\times4096$ spatial grid using the 8th order explicit Runge-Kutta
method described in Section~\ref{sec:num}.  The calculation involved
5400 time steps from $t=0$ to $t=T=0.225$, which took 2.5 days on 12
threads running on a server with two 3.0 GHz Intel Xeon Gold 6136
processors. Additional threads had little effect on the running time
as the FFT calculations require a lot of data movement relative to the
number of floating point operations involved.

Panel (e) of Figure~\ref{fig:ot:2} shows the $\ell^2$ average of the
Fourier mode amplitudes $|\hat\eta_{j_1,j_2}|$ in each shell of
indices satisfying $\max(|j_1|,|j_2|)=s$ for $1\le s\le2048$.  Since
$\hat\eta_{-j_1,-j_2}=\overline{\hat\eta_{j_1,j_2}}$, we can discard
half the modes and sweep through the lattice along straight lines from
$(0,s)$ to $(s,s)$ to $(s,-s)$ to $(1,-s)$, which sweeps out $4s$
index pairs. (The same ordering is used to enumerate the unknowns in
  the nonlinear least squares method proposed in \cite{quasi:trav} to
  compute quasi-periodic traveling water waves.)  We see in panel (e)
that as time increases, the modes continue to decay at an exponential
rate with respect to $s$, but the decay rate is slower at later
times. The rapid dropoff in the mode amplitudes for $s\ge1536$ is due
to the Fourier filter. At the final time $t=T=0.225$, the modes still
decay by 12 orders of magnitude from $s=1$ to $s=1536$, so we believe
the solution is correct to 10--12 digits. A finer grid would be
required to maintain this accuracy over longer times.  As in
  Figure~\ref{fig:periodic:ot}, the additional drop-off in the
  amplitude of Fourier modes from $s=1536$ to $s=2048$ in panel (e) is
  due applying the filter (\ref{eq:filter}) to the solution after
  every timestep.

Beyond monitoring the decay of Fourier modes, as an additional check
of accuracy, we compute the average energy, mass and momentum of the
solution of Figure~\ref{fig:ot:2} as a function of time. The results
are shown in Table~\ref{tbl:energy}. The formulas for $E$, $M$ and
$P_x$ are
\begin{equation}
  \begin{aligned}
    E &= \frac1{(2\pi)^2}\int_{\mathbb{T}^2} \bigg[\,\frac{1}{2} \tilde{\psi}
      \tilde{\varphi}_\alpha
  + \frac{1}{2} g\tilde{\eta}^2\big(1+ \tilde{\xi}_\alpha\big)
  + \tau\left(\sqrt{\big(1+ \tilde{\xi}_\alpha\big)^2 + 
      \big(\tilde{\eta}_\alpha\big)^2}
    - 1\right)\!\bigg]\,d\alpha_1\,d\alpha_2, \\[8pt]
    M &= \frac1{(2\pi)^2}\int_{\mbb T^2}
    \tilde\eta\big(1+\tilde\xi_\alpha\big)\,d\alpha_1\,d\alpha_2, \qquad
    P_x = \frac1{(2\pi)^2}\int_{\mathbb{T}^2} -
    \tilde{\varphi}\tilde{\eta}_\alpha\,
  d\alpha_1\,d\alpha_2,
  \end{aligned}
\end{equation}
which may be shown to be conserved quantites under the free-surface
Euler equations following the derivations in
\cite{zakharov2002new,dyachenko2019}. The only changes required for
the quasi-periodic case are that derivatives and the Hilbert transform
are replaced by their torus versions, and the integrals over $\mbb R$
or $\mbb T$ are replaced by integrals over $\mbb T^2$. We also divide
by $(2\pi)^2$ to obtain average values over the torus. This scaling
has the advantage that $E$, $M$ and $P_x$ do not suddenly jump by a
factor of $2\pi$ when periodic functions are viewed as quasi-periodic
functions that depend on $\alpha_1$ only.  The equations of motion are
Hamiltonian \cite{zakharov1968stability, dyachenko2019} whether or not
the energy is scaled by $1/(2\pi)^2$ --- one just has to multiply the
symplectic 2-form  \cite{abraham:marsden:ratiu} by the same factor.

The numerical results in Table~\ref{tbl:energy} show that energy is
conserved to a relative error of
$1.5\times10^{-13}/0.0975=1.5\times10^{-12}$; mass is conserved to a
relative error of $3.3\times10^{-13}/0.0464=7.1\times10^{-12}$; and
momentum is conserved to an absolute error of $4.6\times10^{-15}$ over
the course of the numerical computation of Figure~\ref{fig:ot:2}. This
gives further evidence that $\tilde\eta$ and $\tilde\varphi$ are
accurate to 10--12 digits.  The mass being negative is an artifact of
the choice of $\xi_1(\sigma)$ and $\eta_1(\sigma)$ in
(\ref{eq:ot:param}). While $\eta_1(\sigma)$ has zero mean as a
function of $\sigma$, its average value with respect to $x$ is
$(2\pi)^{-1}\int_0^{2\pi}\eta_1(\sigma)\xi_1'(\sigma)\,d\sigma
=-(3\pi/40)(\sqrt5-1)=-0.04635254915624$. If we had added a constant to
$\tilde\eta$ to make $M$ zero initially, it would have remained zero
up numerical errors, similar to $P_x$, which is initially zero since
$\tilde\eta$ and $\tilde\varphi$ in (\ref{eq:ot:init:quasi}) are
independent of $\alpha_2$ except for the factor of $\cos(\alpha_2-q)$.
The constant value of the energy would also change if a constant
were added to $\tilde\eta$.

\begin{table}[t]
  \begin{equation*}
\begin{array}{c|c|c|r}
  t & E & M & P_x \hspace*{.375in} \\\hline
  0.000 & 0.097501331570\textbf{54} & -0.046352549156\textbf{43} &  1.19\times10^{-15}\\\hline
  0.025 & 0.097501331570\textbf{57} & -0.046352549156\textbf{23} & -2.98\times10^{-15}\\\hline
  0.050 & 0.097501331570\textbf{58} & -0.046352549156\textbf{21} & -1.24\times10^{-15}\\\hline
  0.075 & 0.097501331570\textbf{57} & -0.046352549156\textbf{19} &  1.15\times10^{-15}\\\hline
  0.100 & 0.097501331570\textbf{59} & -0.046352549156\textbf{17} &  0.43\times10^{-15}\\\hline
  0.125 & 0.097501331570\textbf{60} & -0.046352549156\textbf{16} &  0.34\times10^{-15}\\\hline
  0.150 & 0.097501331570\textbf{62} & -0.046352549156\textbf{14} &  4.59\times10^{-15}\\\hline
  0.175 & 0.097501331570\textbf{63} & -0.046352549156\textbf{13} &  1.35\times10^{-15}\\\hline
  0.200 & 0.097501331570\textbf{64} & -0.046352549156\textbf{11} & -2.18\times10^{-15}\\\hline
  0.225 & 0.097501331570\textbf{69} & -0.046352549156\textbf{10} & -2.10\times10^{-15}\\
\end{array}
\vspace*{10pt}
  \end{equation*}
  \caption{\label{tbl:energy} Average energy, mass and momentum of the
  overturning wave example of Figure~\ref{fig:ot:2} at the times indicated.}
\end{table}

The rationale for setting $q=0.6k\pi$ in (\ref{eq:ot:init:quasi}) is
that $\cos(\alpha_2-q)\approx1$ where the characteristic line
$(\alpha,k\alpha)$ crosses the dropoff in the torus near
$\alpha_1=0.6\pi$ for the first time. Locally,
$\varphi_0(\alpha)=\tilde\varphi_0(\alpha,k\alpha)$ is close to
$\varphi_2(\alpha)$, the initial condition of the auxiliary periodic
problem, so we expect the quasi-periodic wave to evolve similarly to
the periodic wave near $x=\pi$ for a short time. (Here $z=x+iy$
  describes physical space). This is indeed what happens, which may be
seen by comparing panel (b) of Figure~\ref{fig:periodic:ot} to panel
(b) of Figure~\ref{fig:ot:1}, keeping in mind that $t\in[-0.45,0.45]$
in the former plot and $t\in[0,0.225]$ in the latter plot.  Advancing
$\alpha$ from $0.6\pi$ to $10.6\pi$ causes the characteristic line
$(\alpha,k\alpha)$ to cross a periodic image of the dropoff at
$\alpha_2=10.6k\pi$, where
$\cos(\alpha_2-q)=-0.9752\approx-1$. Locally, $\varphi_0(\alpha)$ is
close to $-\varphi_2(\alpha)$, the initial condition of the
time-reversed auxiliary periodic problem.  Thus, we expect the
quasi-periodic wave to evolve similarly to the time-reversed periodic
wave near $x=11\pi$. (Recall that $\xi_0(0.634185\pi)=\pi$, so
  $\xi_0(10.634185\pi)=11\pi$). Comparing panel (b) of
Figure~\ref{fig:periodic:ot} to panel (d) of Figure~\ref{fig:ot:1}
confirms that this does indeed happen.  At most wave peaks, the
velocity potential of the quasi-periodic solution is not closely
related to that of the periodic auxiliary problem since the cosine
factor is not near a relative maximum or minimum, where it is flat.
As a result, the wave peaks of the quasi-periodic solution evolve in
many different ways as $\alpha$ varies over the real line.


\section{Conclusion} 


In this work, we have formulated the two-dimensional, infinite depth
gravity-capillary water wave problem in a spatially quasi-periodic,
conformal mapping framework. We developed two time-stepping strategies
for solving the quasi-periodic initial value problem, an
explicit Runge-Kutta method and an exponential time differencing
scheme. We numerically verified a result in \cite{quasi:trav} that
quasi-periodic traveling waves evolve in time on the torus $\mbb T^2$
in the direction $(1,k)$ without changing form, though their speed is
non-uniform in conformal space if the condition $\tilde\xi(0,0,t)=0$
is imposed via (\ref{eq:C1:opt2}). We then performed a convergence
study to demonstrate the effectiveness of the small-scale
decomposition at removing stiffness from the evolution equations when
the surface tension is large. Finally, we presented the results
  of a large-scale computation of a spatially quasi-periodic
  overturning water wave for which the wave peaks exhibit a wide
  array of dynamic behavior.

In the appendices, we establish minimal conditions to ensure that a
quasi-periodic analytic function $z(w)$ maps the lower half-plane
topologically onto a region bounded above by a curve, and that
$1/|z_w|$ is bounded. We also show that if all the solutions in the
family are single-valued and have no vertical tangent lines, the
corresponding solutions of the original graph-based formulation
(\ref{general_initial})--(\ref{Bernoulli}) of the Euler equations are
quasi-periodic in physical space.  This analysis includes a
  change of variables formula from conformal space to physical space,
  which to our knowledge is also new in the periodic setting. We then
  provide details on implementing the exponential time-differencing
  scheme and discuss generalizations to quasi-periodic water waves
  at the surface of a 3D fluid, where conformal mapping methods are
  no longer applicable.

We believe that spatial quasi-periodicity is a natural setting to
study the dynamics of linear and nonlinear waves, and has largely been
overlooked as a possible third option to the usual modeling assumption
that the solution either evolves on a periodic domain or decays at
infinity. In the future, we plan to develop numerical methods to
compute temporally quasi-periodic water waves with $n=3$ quasi-periods
and to study subharmonic instabilities
\cite{longuet:78,mackay:86,mercer:92,oliveras:11,trichtchenko:16} of
periodic traveling waves and standing waves as well as the
  long-time dynamics of spatially quasi-periodic perturbations.

\appendix

\section{Conformal Mappings of the Lower Half-Plane}
\label{sec:one:one}

In this section we discuss sufficient conditions for an analytic
function $z(w)$ to map the lower half-plane topologically onto a
semi-infinite region bounded above by a parametrized curve. We will
prove the following theorem, and a corollary concerning boundedness
of $1/|z_w|$ when the functions are quasi-periodic.

\begin{theorem}\label{thm:conformal}
  Suppose $\veps>0$ and $z(w)$ is analytic on the half-plane $\mbb
  C^-_\veps = \{w\;:\; \im w<\veps\}$. Suppose there is a constant
  $M>0$ such that $|z(w)-w|\le M$ for $w\in\mbb C^-_\veps$, and that the
  restriction $\zeta=z\vert_\mbb R$ is injective. Then the curve
  $\zeta(\alpha)$ separates the complex plane into two regions, and
  $z(w)$ is an analytic isomorphism of the lower half-plane onto the
  region below the curve $\zeta(\alpha)$.
\end{theorem}

\begin{proof}
We do not assume $\zeta(\alpha)=\xi(\alpha)+i\eta(\alpha)$ is a graph
--- only that it does not self-intersect. We first need to show that
$\Gamma=\{\zeta(\alpha)\,:\,\alpha\in\mbb R\}$ separates the complex
plane into precisely two regions. (In the graph case, this is obvious.)
Let $A=2M$ and consider the linear fractional transformation
\begin{equation}
  T(z) = -\frac{z+Ai}{z-Ai}, \qquad\quad
  T^{-1}(\lambda) = Ai\frac{\lambda-1}{\lambda+1}.
\end{equation}
Note that $T$ maps the real line to the unit circle and
$T^{-1}(e^{-i\theta})=A\tan(\theta/2)$.
Let $g(\theta)=T\circ\zeta\circ T^{-1}(e^{-i\theta})$.  Since
$\zeta(\alpha)$ lies inside a closed ball of radius $M$ centered at
$\alpha$, $\zeta(\alpha)$ remains in the strip $\im z\in[-M,M]$ and
approaches complex $\infty$ as $\alpha\rightarrow\pm\infty$.  Since
$T(\infty)=-1$, $g$ becomes continuous on $[-\pi,\pi]$ if we define
$g(\pm\pi)=-1$.  Since $T$ is bijective and $\zeta$ is injective, $g$
is a Jordan curve and separates the complex plane into two regions.
The curve $g(\theta)$ takes values in the set $\{\lambda :
|\lambda+1/3|\ge2/3, |\lambda-1|\le 2\}$, which is the image of the
strip $\im z\in[-M,M]$ under $T$. An argument similar to Lemma~2 of
section 4.2.1 of \cite{ahlfors} shows that if $|\lambda+1/3|<2/3$, then
$\lambda$ is inside the Jordan curve. In particular, if $\im z<-M$,
then $T(z)$ is inside the curve.  We conclude that there is a
well-defined ``fluid'' region $\Omega$ that is mapped topologically
by $T$ to the inside of the Jordan curve, and a ``vacuum'' region
$\Omega_v$ that is mapped topologically to the outside of the Jordan
curve.

Next we show that $z(w)$ is univalent and maps the lower half-plane
onto the fluid region.  Consider the path $\gamma_a$ in the $w$-plane
that traverses the boundary $\pa S_a$ of the half-disk
$S_a=\{w\,:\,|w|<a,\,\im w<0\}$. Suppose
$z_0\in\Omega\cup\Omega_v=\mbb C\setminus\Gamma$.  If $a>|z_0|+M$,
then $z(w)-z_0$ has no zeros on $\gamma_a$. Indeed, $w\in\gamma_a$
requires $\im w=0$ or $|w|=a$. Assuming $z(w)=z_0$ would require
$z_0\in\Gamma$ or $|z(w)|\ge a-M>|z_0|$, a contradiction in either
case.  We can therefore define the winding number of
$\Gamma_a=z(\gamma_a)$ around $z_0$,
\begin{equation}\label{eq:nGam}
  n(\Gamma_a,z_0) = \int_{\Gamma_a}\frac{dz}{z-z_0} =
  \int_{\gamma_a}\frac{z'(w)\,dw}{z(w)-z_0}, \qquad (a>|z_0|+M).
\end{equation}
Since $n(\Gamma_a,z_0)$ counts the number of solutions of $z(w)=z_0$
inside $S_a$ and all solutions in the lower half-plane belong to $S_a$
as soon as $a>|z_0|+M$, $n(\Gamma_a,z_0)$ is a non-negative integer
that gives the number of solutions of $z(w)=z_0$ in the lower
half-plane. It is independent of $a$ once $a>|z_0|+M$, which
is assumed in (\ref{eq:nGam}).

We decompose $n(\Gamma_a,z_0) = n_2(z_0,a) - n_1(z_0,a)$, where
\begin{equation}\label{eq:n1:n2}
  n_1(z_0,a) = \frac1{2\pi i}\int_{-a}^a
  \frac{\zeta'(\alpha)d\alpha}{\zeta(\alpha)-z_0}, \qquad\quad
  n_2(z_0,a) = \frac1{2\pi i}\int_{-\pi}^0
  \frac{z'(ae^{i\theta})(iae^{i\theta})}{z(ae^{i\theta})-z_0}\,d\theta.
\end{equation}
Let $n_1(z_0) = \lim_{a\rightarrow\infty}n_1(z_0,a)$.  We will show that
\begin{equation}\label{eq:n1:formula}
  n_1(z_0) = 
  \begin{cases}
    -1/2, & z_0\in\Omega, \\
    \phm 1/2, & z_0\in\Omega_v.
  \end{cases}
\end{equation}
First, if $z_0=-iA$, where $A>M$, then $\im\{\zeta(\alpha)-z_0\}>0$
for $\alpha\in\mbb R$.  Thus, $\zeta(\alpha)-z_0$ does not cross the
principal branch cut of the logarithm. As a result,
\begin{equation*}
  n_1(-iA,a) = \frac1{2\pi i}\left[
    \ln\left|\frac{\zeta(a)+iA}{\zeta(-a)+iA}\right| +
    i\opn{Arg}\big(\zeta(a)+iA\big) -
    i\opn{Arg}\big(\zeta(-a)+iA\big)\right].
\end{equation*}
Since $|\zeta(\alpha)-\alpha|\le M<A$, as $a\rightarrow\infty$ we obtain
\begin{equation*}
  n_1(-iA) = \frac{\ln(1)+i0-i\pi}{2\pi i} = -\frac12.
\end{equation*}
A similar argument shows that $n_1(iA)=1/2$. Let $[z_0,z_1]$ denote
the line segment in $\mbb C$ connecting two points $z_0$ and $z_1$,
and suppose this line segment does not intersect the curve
$\zeta(\alpha)$. We claim that if the limit defining $n_1(z_0)$
exists, the limit defining $n_1(z_1)$ also exists, and
$n_1(z_1)=n_1(z_0)$. First note that
\begin{equation*}
  n_1(z_1,a)=n_1(z_0,a)+\frac1{2\pi i}\int_{-a}^a\left[
    \frac{\zeta'(\alpha)}{\zeta(\alpha)-z_1}-
    \frac{\zeta'(\alpha)}{\zeta(\alpha)-z_0}\right]
  \,d\alpha.
\end{equation*}
Since $\zeta(\alpha)$ does not cross $[z_0,z_1]$, 
$\frac{\zeta(\alpha)-z_1}{\zeta(\alpha)-z_0}$ does not cross
the negative real axis or 0. Thus,
\begin{equation*}
  n_1(z_1,a)=n_1(z_0,a)+\left[\opn{Log}\frac{\zeta(\alpha)-z_1}{
      \zeta(\alpha)-z_0}\right]_{\alpha=-a}^{\alpha=a}.
\end{equation*}
Taking the limit as $a\rightarrow\infty$ gives $n_1(z_1)=n_1(z_0)+0$,
as claimed. Every point of $\Omega$ is connected to $-iA$ by a
polygonal path that remains inside $\Omega$, and every point in
$\Omega_v$ is connected to $iA$ by a polygonal path that remains
inside $\Omega_v$. The result (\ref{eq:n1:formula}) follows.

Next consider $n_2(z_0,a)$ in (\ref{eq:n1:n2}).
Since $|z(w)-w|\le M$, the Cauchy integral formula gives
$|z'(\alpha+i\beta)-1|\le
M/(\veps-\beta)$ for $\beta\in(-\infty,\veps)$.
For $a>2(M+|z_0|)$,
\begin{equation*}
  |z(ae^{i\theta})-z_0|\ge (a-M-|z_0|)> a/2,
\end{equation*}
thus the modulus of the integrand in the formula for
  $n_2(z_0,a)$ in (\ref{eq:n1:n2}) is bounded
uniformly by $2(1+M/\veps)$ for large $a$.  For fixed
$\theta\in(-\pi,0)$, $z'(ae^{i\theta})\rightarrow 1$ and
$[z(ae^{i\theta})-z_0]/a\rightarrow e^{i\theta}$, so the integrand
approaches $i$ pointwise on the interior of the integration interval
as $a\rightarrow\infty$. By the dominated convergence theorem,
$\lim_{a\rightarrow\infty}n_2(z_0,a)=1/2$. Combining these results, we
find that $\lim_{a\rightarrow\infty} n(\Gamma_a,z_0) =
1/2-n_1(z_0)$. But since $n(\Gamma_a,z_0)$ is constant for
$a>|z_0|+M$, we conclude that
\begin{equation}
  n(\Gamma_a,z_0) = \left\{\begin{array}{rl}
    1, & z_0\in\Omega \\
    0, & z_0\in\Omega_v
  \end{array}\right\}, \qquad\quad
  (a>|z_0|+M).
\end{equation}
This shows that when solving the equation $z(w)=z_0$, if
$z_0\in\Omega$, there is precisely one solution $w_0$ in the lower
half-plane, and if $z_0\in\Omega_v$, there are no solutions $w_0$ in
the lower half-plane. Since $z(w)$ is an open mapping, it cannot map a
point in the lower half-plane to the boundary $\Gamma$, since a nearby
point would then have to be mapped to $\Omega_v$. It follows that
$z(w)$ is a 1-1 mapping of $\{\im w<0\}$ onto $\Omega$. It is then a
standard result that $z'(w)$ has no zeros in the lower half-plane and
the inverse function $w(z)$ exists and is analytic on $\Omega$.
\end{proof}

\begin{example}
  The function $z(w)=2w^2/(2w-i)$ satisfies $|z(w)-w|=|iw|/|2w-i|$.
  Writing $w=\alpha+i\beta$, we have
  $|z(w)-w|^2=(\alpha^2+\beta^2)/[4\alpha^2+(2\beta-1)^2]$.  If
  $\beta\in(-\infty,1/4]$, then $(2\beta-1)^2=4\beta^2-4\beta+1\ge4\beta^2$ and
  $|z(w)-w|^2\le1/4$. Moreover,
  $\zeta(\alpha)=2\alpha^2(2\alpha+i)/(4\alpha^2+1)$ is injective in
  spite of a cusp at the origin (where $\zeta'(0)=0$).  The hypotheses
  of Theorem~\ref{thm:conformal} are satisfied with $\veps=1/4$ and
  $M=1/2$, so $z$ maps the half-plane $\mbb C^-$ conformally onto the
  region below the curve $\zeta(\alpha)$, in spite of the cusp. We
  will usually assume $\zeta'(\alpha)\ne0$ for $\alpha\in\mbb R$ so
  that the curve is smooth.
\end{example}

\begin{cor}\label{cor:conformal}
  Suppose $k>0$ is irrational,
  $\tilde\eta(\alpha_1,\alpha_2)=\sum_{(j_1,j_2)\in\mbb Z^2}
  \hat\eta_{j_1,j_2}e^{i(j_1\alpha_1+j_2\alpha_2)}$, and there exist constants $C$
          and $\veps>0$ such that
  \begin{equation}
    \hat\eta_{-j_1,-j_2}=\overline{\hat\eta_{j_1,j_2}}, \qquad
    \big|\hat\eta_{j_1,j_2}\big|\le Ce^{-3\veps K\max(|j_1|,|j_2|)},
    \qquad\quad
    (j_1,j_2)\in\mbb Z^2,
  \end{equation}
  where $K=\max(k,1)$.  Let $x_0$ be real and define
  $\tilde\xi=x_0+H[\tilde\eta]$, $\tilde\zeta= \tilde\xi+i\tilde\eta$
  and
  \begin{equation}
    \tilde z(\alpha_1,\alpha_2,\beta) = x_0 + i\hat\eta_{0,0} +
    \sum_{j_1+j_2k<0} 2i\hat\eta_{j_1,j_2}e^{-(j_1+j_2k)\beta}e^{i(j_1\alpha_1+
        j_2\alpha_2)}, \qquad (\beta<\veps),
  \end{equation}
  where the sum is over all integer pairs $(j_1,j_2)$ satisfying the
  inequality.
  Suppose also that for each fixed $\theta\in[0,2\pi)$, the function
  $\alpha\mapsto\zeta(\alpha;\theta)= \alpha+\tilde
  \zeta(\alpha,\theta+k\alpha)$ is injective from $\mbb R$ to $\mbb C$
  and $\zeta_\alpha(\alpha;\theta)\ne0$ for $\alpha\in\mbb R$. Then
  for each $\theta\in\mbb R$, the curve $\zeta(\alpha;\theta)$
  separates the complex plane into two regions and
  \begin{equation}
    z(\alpha+i\beta;\theta) = (\alpha+i\beta) +
    \tilde z(\alpha,\theta+k\alpha,\beta), \qquad
    (\beta<\veps)
  \end{equation}
  is an analytic isomorphism of the lower half-plane onto the
  region below $\zeta(\alpha;\theta)$. Moreover, there is a constant
  $\delta>0$ such that $|z_w(w;\theta)|\ge\delta$ for
  $\im w\le 0$ and $\theta\in\mbb R$.
\end{cor}

\begin{proof}
  First we confirm that $z(w;\theta)$ and $\zeta(w;\theta)$ satisfy the
  hypotheses of Theorem~\ref{thm:conformal}. The formula
  \begin{equation}\label{eq:z:w:cor}
    z(w;\theta) = w + x_0 + i\hat\eta_{0,0} +
    \sum_{j_1+j_2k<0}\left(2i\hat\eta_{j_1,j_2}
      e^{ij_2\theta}\right)e^{i(j_1+j_2k)w}
  \end{equation}
  expresses $z(w;\theta)$ as a uniformly convergent series of
  analytic functions on the region $\im w<\veps$, so it is analytic in
  this region. This follows from the inequalities
  \begin{equation}
    \begin{gathered}
      0 < -(j_1+j_2k)\le|j_1|+|j_2|k\le 2K\max(|j_1|,|j_2|), \\
      -(j_1+j_2k)\beta \le -(j_1+j_2k)\veps \le
      2\veps K\max(|j_1|,|j_2|), \qquad (-\infty<\beta\le\veps) \\
      \left|\left(2i\hat\eta_{j_1,j_2}
        e^{ij_2\theta}\right)e^{i(j_1+j_2k)w}\right|
      \le 2Ce^{-\veps K\max(|j_1|,|j_2|)}, \qquad (\im w\le\veps)
    \end{gathered}
  \end{equation}
  and the fact that for each non-negative integer $s$, there are $4s$
  index pairs $(j_1,j_2)$ in the shell $\max(|j_1|,|j_2|)=s$ and
  satisfying $j_1+j_2k<0$:
  \begin{equation}
      \sum_{j_1+j_2k<0}\left|\left(2i\hat\eta_{j_1,j_2}
        e^{ij_2\theta}\right)e^{i(j_1+j_2k)w}\right| \le
      \sum_{s=1}^\infty (2C)(4s)e^{-\veps Ks}<\infty, \qquad
      (\im w\le\veps).
  \end{equation}
  This also implies that there is a bound $M$ such that
  $|z(w;\theta)-w|\le M$ for $\im w\le\veps$ and $\theta\in\mbb R$.  Let
  $\xi(\alpha;\theta)=\alpha+\tilde\xi(\alpha,\theta+k\alpha)$ and
  $\eta(\alpha;\theta)=\tilde\eta(\alpha,\theta+k\alpha)$ denote the
  real and imaginary parts of $\zeta(\alpha;\theta)$.  Setting
  $w=\alpha\in\mbb R$ in (\ref{eq:z:w:cor}) and taking real and
  imaginary parts confirms that $z(w;\theta)\vert_{w=\alpha} =
  \zeta(\alpha;\theta)$.  By assumption, $\zeta(\alpha;\theta)$ is
  injective, so Theorem~\ref{thm:conformal} implies that $z(w;\theta)$
  is an analytic isomorphism of the lower half-plane onto the region
  below the curve $\zeta(\alpha;\theta)$.  Differentiating
  (\ref{eq:z:w:cor}) term by term \cite{ahlfors} shows that
  $z_w(\alpha+i\beta;\theta)=F(\alpha,\theta+k\alpha,\beta)$, where
  \begin{equation}
    F(\alpha_1,\alpha_2,\beta) = 1 - \sum_{j_1+j_2k<0}
    2(j_1+j_2k)\hat\eta_{j_1,j_2}e^{-(j_1+j_2k)\beta}e^{i(j_1\alpha_1+j_2\alpha_2)}.
  \end{equation}
  We claim that $F(\alpha_1,\alpha_2,\beta)\rightarrow1$ uniformly
  in $(\alpha_1,\alpha_2)$ as $\beta\rightarrow-\infty$.
  Indeed, arguing as in (\ref{eq:zw:lim}), we see that
    $|F(\alpha_1,\alpha_2,\beta) - 1|=
    |z_w(\alpha_1+i\beta;\alpha_2-k\alpha_1)-1|\le M/(\veps-\beta)$.
    Thus, for $\beta\le -B$ with $B=2M$,
  $|F(\alpha_1,\alpha_2,\beta)|\ge1/2$.  Since
  $|F(\alpha_1,\alpha_2,\beta)|$ is continuous, it achieves its
  minimum over $(\alpha_1,\alpha_2)\in\mbb T^2$ and $-B\le\beta\le0$.
  Denote this minimum by $\delta$. If $\delta$ were zero, there would
  exist $\alpha_1$, $\alpha_2$ and $\beta\le0$ such that
  $F(\alpha_1,\alpha_2,\beta)=0$. But then
  $z_w(\alpha_1+i\beta;\theta)=F(\alpha_1,\theta+k\alpha_1,\beta)=0$
  with $\theta=\alpha_2-k\alpha_1$. The case $\beta=0$ is ruled out by
  the assumption that $\zeta_\alpha(\alpha;\theta)\ne0$ while
  $\beta<0$ contradicts $z(w;\theta)$ being 1-1 on $\mbb C^-$.  So
  $\delta>0$ and $|F(\alpha_1,\alpha_2,\beta)|\ge\min(\delta,1/2)$ for
  all $\beta\le0$. Decreasing $\delta$ to 1/2 if necessary gives the
  desired lower bound $|z_w(w;\theta)|\ge\delta$.
\end{proof}

\section{Quasi-Periodic Families of Solutions}\label{sec:families}

In this appendix we explore the effect of introducing phases in the
reconstruction of one-dimensional quasi-periodic solutions of
(\ref{general_conformal}) from solutions of the torus version of these
equations. This ultimately makes it possible to show that if all the
solutions in the family are single-valued and have no vertical tangent
lines, the corresponding solutions of the original graph-based
formulation (\ref{general_initial})--(\ref{Bernoulli}) of the Euler
equations are quasi-periodic in physical space.

\begin{theorem}\label{thm:family}
  The solution pair $(\tilde\zeta,\tilde\varphi)$ on the torus
  represents an infinite
  family of quasi-periodic solutions on $\mbb R$ given by
  \begin{equation}\label{eq:family}
    \begin{aligned}
      \zeta(\alpha,t\,;\,\theta_1,\theta_2,\delta) &= \alpha +
      \delta + \tilde\zeta(
      \theta_1+\alpha,\theta_2+k\alpha,t), \\
    \varphi(\alpha,t\,;\,\theta_1,\theta_2) &= \tilde\varphi(
      \theta_1+\alpha,\theta_2+k\alpha,t),
    \end{aligned} \qquad
    \left( \begin{aligned}
        \alpha\in\mbb R, \, t\ge0 \\
        \theta_1,\theta_2,\delta\in\mbb R
        \end{aligned}\right).
  \end{equation}
\end{theorem}

\begin{proof}
  We claim that by solving (\ref{general_conformal}) throughout $\mbb
  T^2$ in the sense of Remark~\ref{rmk:torus}, any one-dimensional
  (1D) slice of the form (\ref{eq:family}) will satisfy the kinematic
  condition (\ref{kinematic_conformal_1}) and the Bernoulli equation
  (\ref{eq:bernoulli:conf}). Let us freeze $\theta_1$, $\theta_2$ and
  $\delta$ and drop them from the notation on the left-hand side of
  (\ref{eq:family}). Consider substituting $\eta=\im\zeta$ and
  $\varphi$ from (\ref{eq:family}) into (\ref{general_conformal}), and
  let $u(\alpha)=\tilde u(\theta_1+\alpha,\theta_2+k\alpha)$ represent
  the input of any $\alpha$-derivative or Hilbert transform in an
  intermediate calculation.  Both $\eta$ and $\varphi$ are of this
  form.  By Remark~\ref{rmk:f:quasi},
  $H[u](\alpha) = H[\tilde u](\theta_1+\alpha,\theta_2+k\alpha)$, and
  clearly $u'(\alpha) = [(\pa_{\alpha_1}+k\pa_{\alpha_2})\tilde u]
  (\theta_1+\alpha,\theta_2+k\alpha)$, so the output retains this
  form. We conclude that computing (\ref{general_conformal}) on the
  torus gives the same results for $\tilde\eta_t$ and
  $\tilde\varphi_t$ when evaluated at
  $(\theta_1+\alpha,\theta_2+k\alpha)$ as the 1D calculations of
  $\eta_t$ and $\varphi_t$ when evaluated at~$\alpha$. Since
  $\tilde\xi(\cdot,t)=x_0(t)+H[\tilde\eta(\cdot,t)]$ on $\mbb T^2$,
  \begin{equation}\label{eq:xi:family}
    \xi(\alpha,t) = \alpha+\delta+x_0(t)+H[\eta(\cdot,t)](\alpha),
  \end{equation}
  which follows from (\ref{eq:family}) and $H[\eta(\cdot,t)](\alpha)=
  H[\tilde\eta(\cdot,t)](\theta_1+\alpha,\theta_2+k\alpha)$.  Thus,
  computing $\xi_\alpha=1+H[\eta_\alpha]$ in (\ref{general_conformal})
  gives the same result as just differentiating $\xi$ from
  (\ref{eq:family}) and (\ref{eq:xi:family}).  In the 1D problem, the
  right-hand side of (\ref{kinematic_conformal_3}) represents complex
  multiplication of $z_\alpha$ with a bounded analytic function
  (namely $z_t/z_\alpha$) whose imaginary part equals $-\chi$ on the
  real axis; thus, in (\ref{kinematic_conformal_3}), $\xi_t$ differs
  from $H[\eta_t]$ by a constant. This constant is determined by
  comparing $\xi_t$ in (\ref{kinematic_conformal_3}) with $\xi_t$ from
  (\ref{eq:xi:family}), which leads to the same formula
  (\ref{eq:x0:evol}) for $dx_0/dt$ that is used in the torus
  calculation. Here we note that a phase shift does not affect the
  mean of a periodic function on the torus,
  i.e.~$P_0[S_{\bds\theta}\tilde u]=P_0[\tilde u]$ where
  $S_{\bds\theta}[\tilde u](\bds\alpha)=\tilde u(\bds\alpha+\bds\theta)$.
  We have assumed
  that in the 1D calculation, $C_1$ is chosen to agree with that of
  the torus calculation. Since $C_1$ only affects the tangential
  velocity of the interface parametrization, it can be specified
  arbitrarily. Left-multiplying (\ref{kinematic_conformal_3}) by
  $(-\eta_\alpha, \xi_\alpha)$ eliminates $C_1$ and yields the
  kinematic condition (\ref{kinematic_conformal_1}). Since the
  Bernoulli equation (\ref{eq:bernoulli:conf}) holds on the torus, it
  also holds in the 1D calculation, as claimed.
\end{proof}

For each solution in the family (\ref{eq:family}), there are many
others that represent identical dynamics up to a spatial phase shift
or $\alpha$-reparametrization. Changing $\delta$ merely shifts the
solution in physical space. In fact, $\delta$ does not appear in the
equations of motion (\ref{kinematic_conformal_3}) --- it is only used
to reconstruct the curve via (\ref{eq:xi:family}).  The relations
\begin{equation}\label{eq:canon:fam}
  \begin{aligned}
    \zeta(\alpha+\alpha_0,t\,;\,\theta_1,\theta_2,\delta) &= \zeta(
      \alpha,t\,;\,\theta_1+\alpha_0, \theta_2+k\alpha_0,\delta+\alpha_0), \\
    \varphi(\alpha+\alpha_0,t\,;\,\theta_1,\theta_2) &= \varphi(
      \alpha,t\,;\,\theta_1+\alpha_0, \theta_2+k\alpha_0),
  \end{aligned}
\end{equation}
show that shifting $\alpha$ by $\alpha_0$ leads to another
solution already in the family. This shift reparametrizes the curve
but has no effect on its evolution in physical space. If we identify
two solutions that differ only by a spatial phase shift or
$\alpha$-reparametrization, the parameters
$(\theta_1,\theta_2,\delta)$ become identified with
$(0,\theta_2-k\theta_1,0)$.  Every solution is therefore equivalent
to one of the form
\begin{equation}\label{eq:smaller:family}
  \begin{aligned}
    \zeta(\alpha,t\,;\,0,\theta,0) &=
    \alpha + \tilde\zeta(\alpha,\theta+k\alpha,t), \\
    \varphi(\alpha,t\,;\,0,\theta) &=
    \tilde\varphi(\alpha,\theta+k\alpha,t)
  \end{aligned}
  \qquad
  \alpha\in\mbb R\,,\;
  t\ge 0\,,\;
  \theta\in[0,2\pi).
\end{equation}
Within this smaller family, two values of $\theta$ lead to
equivalent solutions if they differ by $2\pi(n_1k+n_2)$ for some
integers $n_1$ and $n_2$. This equivalence is due to
solutions ``wrapping around'' the torus with a spatial shift,
\begin{equation}\label{eq:wrap:around}
  \zeta(\alpha+2\pi n_1,t\,;\, 0,\theta,0) =
  \zeta(\alpha,t\,;\,0,\theta+2\pi(n_1k+n_2),2\pi n_1), \quad
  \big(\alpha\in[0,2\pi),\;n_1\in\mbb Z\big).
\end{equation}
Here $n_2$ is chosen so that $0\le\big(\theta+2\pi(n_1k+n_2)\big)<2\pi$ and we
used periodicity of $\zeta(\alpha,t\,;\,\theta_1,\theta_2,\delta)$
with respect to $\theta_1$ and $\theta_2$. It usually suffices to
restrict attention to $\alpha\in[0,2\pi)$ by making use of
(\ref{eq:wrap:around}).  One exception is determining whether the
curve self-intersects. In that case it is more natural to tile the
plane with periodic copies of the torus and consider the straight line
parametrization of (\ref{eq:smaller:family}). Indeed, it is
conceivable that
\begin{equation}\label{eq:intersect}
  \zeta(\alpha,t;0,\theta,0)=\zeta(\beta,t;0,\theta,0)
\end{equation}
with $|\alpha-\beta|$ as large as $2M$, where $M$ is a bound on
$|\tilde\zeta|$ over $\mbb T^2$, and the condition
(\ref{eq:intersect}) becomes hard to understand if
(\ref{eq:wrap:around}) is used to map $\alpha$ and $\beta$ back to
$[0,2\pi)$ with different choices of $n_1$ or $n_2$.

We now show that $\eta^\ph(x,t\,;\,\theta_1,\theta_2,\delta)$ and
$\varphi^\ph(x,t\,;\,\theta_1,\theta_2,\delta)$ can be defined and
computed easily from $\zeta(\alpha,t\,;\,\theta_1,\theta_2,\delta)$ and
$\varphi(\alpha,t\,;\,\theta_1,\theta_2)$ if all of the waves in the
family (\ref{eq:smaller:family}) are single-valued and have no
vertical tangent lines, and that $\eta^\ph$ and $\varphi^\ph$ are
quasi-periodic functions of~$x$. To simplify notation, let
  $\bds\alpha=(\alpha_1,\alpha_2)$, $\bds x=(x_1,x_2)$ and
  $\bds k=(1,k)$.

\begin{theorem}\label{thm:quasi:phys1}
Fix $t\ge0$ and suppose $\xi_\alpha(\alpha,t\,;\,0,\theta,0)>0$ for all
$(\alpha,\theta)\in[0,2\pi)\times[0,2\pi)$. Then the equation
\begin{equation}\label{eq:mcA:def1}
  \mc A(\bds x,t) + \tilde\xi\big(\bds x+\bds k\mc A(\bds x,t)\,,\,t\big) = 0,
\end{equation}
defines a unique function $\mc A(x_1,x_2,t)$ on $\mbb T^2$ that is
periodic and real analytic in $x_1$ and $x_2$. The inverse of the
change of variables $\bds x = \bds\alpha + \bds k\tilde\xi(\bds\alpha,t)$
on $\mbb T^2$ is given by
  \begin{equation}\label{eq:chg:vars:T2}
    \bds\alpha = \bds x + \bds k\mc A(\bds x,t).
  \end{equation}
\end{theorem}

\begin{proof}
  First we check that if $\mc A$ satisfies (\ref{eq:mcA:def1}), then
  (\ref{eq:chg:vars:T2}) is the inverse of the change of variables
  $\bds x = \bds\alpha + \bds k\tilde\xi(\bds\alpha,t)$. Given
  $\bds x\in\mbb T^2$, define $\bds\alpha$ by (\ref{eq:chg:vars:T2}).
  Then
  \begin{equation}
    \bds\alpha + \bds k\tilde\xi(\bds\alpha,t) =
    \big(\bds x + \bds k\mc A(\bds x,t)\big) - \bds k\mc A(\bds x,t) = \bds x,
  \end{equation}
  as required. Next we show existence and uniqueness of a solution
  $\mc A$ of (\ref{eq:mcA:def1}) under the assumed hypotheses. Given
  $\bds\alpha=(\alpha_1,\alpha_2)\in\mbb T^2$, the definition
  (\ref{eq:family}) gives
  \begin{equation}\label{eq:1plusD}
    \xi_\alpha(\alpha_1,t\,;\,0,\alpha_2-k\alpha_1,0)
       = 1+[\pa_{\alpha_1}+k\pa_{\alpha_2}]\tilde\xi(\alpha_1, \alpha_2, t),
  \end{equation}
  where the left-hand side means $(d/d\alpha)\big
  \vert_{\alpha=\alpha_1} \xi(\alpha,t\,;\,0,\alpha_2-k\alpha_1,0)$.
  We know the right-hand side is periodic and continuous on $\mbb T^2$
  while the left-hand is positive on the primitive cell
  $\{(\alpha_1,\alpha_2): 0\le\alpha_1<2\pi\;,\;k\alpha_1\le
  \alpha_2<k\alpha_1+2\pi\}$.  Therefore, both sides of
  (\ref{eq:1plusD}) are bounded below by some $\veps(t)>0$ that does
  not depend on $\bds\alpha\in\mbb T^2$. Let $M(t)$ be a bound on
  $|\tilde\xi(\bds\alpha,t)|$ over $\mbb T^2$. Then for fixed
  $\bds x\in\mbb R^2$ (with $t$ also fixed), the function
  $g(\alpha) = g(\alpha;\bds x,t) = \alpha + \tilde\xi(\bds
    x+\bds k \alpha,t)$ is strictly monotonically increasing on $\mbb R$ (as
    $g'(\alpha)\ge\veps(t)$) and satisfies
  $g(-M(t))\le0$ and $g(M(t))\ge0$. Thus, we can define $\mc A(\bds x,t)$ as
  the unique solution of $g(\alpha)=0$. It follows that $|\mc A(\bds
    x,t)|\le M(t)$. If $n_1$ and $n_2$ are integers, replacing $\bds x$
  in (\ref{eq:mcA:def1}) by $\bds y=(y_1,y_2)=(x_1+2\pi n_1,x_2+2\pi
    n_2)$ and using periodicity of $\tilde\xi(\bds\alpha,t)$ gives
  \begin{equation}
    \mc A(\bds y,t) + \tilde\xi\big(\bds x + \bds k\mc A(\bds y,t),t\big) = 0.
  \end{equation}
  Since the solution of this equation is unique, $\mc A(\bds y,t)=\mc
  A(\bds x,t)$. This shows that $\mc A(\bds x,t)$ is periodic in $\bds
  x$, and hence well-defined on $\mbb T^2$. It is also real analytic,
  which follows from the implicit function theorem, noting that
  $g(\alpha;x_1,x_2,t)$ is real analytic in $\alpha$, $x_1$ and
  $x_2$ for fixed $t$ and $\pa g/\pa \alpha$ is never zero. For the
  same reason, $\mc A(\bds x,t)$ will depend as smoothly on $t$ as
  $\tilde\xi(\bds\alpha,t)$ does.
\end{proof}

The change of variables (\ref{eq:chg:vars:T2}) allows us transform the
torus functions $\tilde\xi$, $\tilde\eta$ and $\tilde\varphi$ in
conformal space to physical space
\begin{equation}
  \begin{aligned}
    \tilde\eta^\ph(\bds x,t) &= \tilde\eta(\bds x + \bds k\mc A(\bds x,t),\,t), \\
    \tilde\varphi^\ph(\bds x,t) &= \tilde\varphi(\bds x + \bds k\mc A(\bds x,t),\,t),
  \end{aligned} \qquad\quad
  \begin{aligned}
    \tilde\eta(\bds\alpha,t) &= \tilde\eta^\ph
              (\bds\alpha + \bds k\tilde\xi(\bds\alpha,t),\,t), \\
    \tilde\varphi(\bds\alpha,t) &= \tilde\varphi^\ph
                 (\bds\alpha + \bds k\tilde\xi(\bds\alpha,t),\,t).
  \end{aligned}
\end{equation}
We then write $\bds\theta=(\theta_1,\theta_2)$ and define the quasi-periodic slices
\begin{equation}
  \begin{aligned}
    \eta^\ph(x,t\,;\,\bds\theta,\delta) &=
    \tilde\eta^\ph\big(\bds\theta+\bds k(x-\delta)\,,\,t\big), \\
    \varphi^\ph(x,t\,;\,\bds\theta,\delta) &=
    \tilde\varphi^\ph\big(\bds\theta+\bds k(x-\delta)\,,\,t\big),
  \end{aligned}
\end{equation}
which express $\zeta(\alpha,t\,;\,\bds\theta,\delta)$ as a graph and
$\varphi(\alpha,t\,;\,\bds\theta)$ as a function of $x$:
\begin{align}
\notag
\eta(\alpha,t\,;\,\bds\theta,\delta) &= \tilde\eta(\bds\theta+\bds k\alpha,\,t) \\
\label{eq:eta:tilde:fam}
&=  \tilde\eta^\ph\big(\bds\theta + \bds k\alpha +
    \bds k\tilde\xi(\bds\theta+\bds k\alpha,t),\,t\big) \\
  \notag
  &= \tilde\eta^\ph\big(\bds\theta + \bds k(\xi(\alpha)-\delta),\,t\big) =
  \eta^\ph\big(\xi(\alpha),t\,;\,\bds\theta,\delta\big), \\
  \label{eq:phi:tilde:fam}
  \varphi(\alpha,t\,;\,\bds\theta) &=
  \tilde\varphi^\ph(\bds\theta+\bds k(\xi(\alpha)-\delta),\,t) =
  \varphi^\ph\big(\xi(\alpha),t\,;\,\bds\theta,\delta\big),
\end{align}
where $\xi(\alpha)=\xi(\alpha,t\,;\,\bds\theta,\delta)=
\alpha+\delta+\tilde\xi(\bds\theta+\bds k\alpha,t)$.
These equations confirm (\ref{eta_chain_rule}) and (\ref{eq:phi:tilde}),
which are the assumptions connecting solutions of
(\ref{general_initial}) to those of (\ref{general_conformal}). Thus,
$\eta^\ph(x,t\,;\,\bds\theta,\delta)$ and
$\varphi^\ph(x,t\,;\,\bds\theta,\delta)$ are solutions of
(\ref{general_initial}), the graph-based formulation of the water
wave equations. In the right-hand sides of (\ref{eq:eta:tilde:fam})
and (\ref{eq:phi:tilde:fam}), we can compute the $\alpha$ such
that $\xi(\alpha)=x$ as follows:
\begin{equation}
  \begin{aligned}
    \xi(\alpha,t\,;\,\bds\theta,\delta)=x \quad &\Leftrightarrow \quad
    \alpha+\delta + \tilde\xi(\bds\theta+\bds k\alpha,t) = x \\
    & \Leftrightarrow
    \quad (\bds\theta + \bds k\alpha) + \bds k\tilde\xi(\bds\theta + \bds k\alpha,t)
      = [\bds\theta + \bds k(x-\delta)] \\
      & \Leftrightarrow \quad
      (\bds\theta + \bds k\alpha) = [\bds\theta + \bds k(x-\delta)] + \bds k
      \mc A\big(\bds\theta + \bds k(x-\delta),t\big) \\
    & \Leftrightarrow \quad \alpha = (x-\delta) + \mc A(\bds\theta + \bds k(x-\delta),t),
  \end{aligned}
\end{equation}
where we used (\ref{eq:chg:vars:T2}) with $\bds\alpha = \bds\theta + \bds k\alpha$ and
$\bds x=\bds\theta + \bds k(x-\delta)$
to obtain the third line from the second.

\section{Details on implementing the exponential time differencing schemes}
\label{sec:etd}

In this section we summarize how to solve the evolution equations
(\ref{eq:ssd:L})--(\ref{eq:ssd:N}) using the 4-stage 4th order ETD scheme of
Cox and Matthews \cite{cox:matthews, kassam} or the 6-stage 5th order ETD
scheme of Whalen, Brio and Moloney \cite{whalen:etd}.  When using an
$s$-stage ETD scheme to solve the ODE
\begin{equation}\label{eq:etd:gen}
  u_t = Lu+\mc{N}(t,u),
\end{equation}
the numerical solution $u_n$ is advanced from $t_n$ to $(t_n+h)$
via
\begin{equation}\label{eq:ETD:step}
  \begin{aligned}
    &\;\; \left.\begin{aligned}
       U_r &= e^{c_rhL}u_n + h\sum_{j=1}^s a_{rj}(hL)\mc N_j, \\[-3pt]
       \mc N_r &= \mc N(t_n+c_rh,U_r), \\
     \end{aligned}\right\} \;\; 1\le r\le s, \\
    & u_{n+1} = e^{hL}u_n + h\sum_{r=1}^s b_r(hL)\mc N_r.
  \end{aligned}
  \;\;
  \parbox{1.5in}{
    \centering
    $\begin{array}{c|c} c & A(z) \\ \hline & \raisebox{-2pt}{
        $b(z)^T$} \end{array}$ \\[2pt]
      Butcher array
  }
\end{equation}
The Butcher array for the Cox-Matthews scheme may be written
\begin{equation}
  \begin{array}{c|cccc}
    0 & 0 \\
    \frac12 & \frac12\phi_1(\frac z2) & 0 \\[2pt]
    \frac12 & 0 & \frac12\phi_1(\frac z2) & 0 \\[2pt]
    1 & \frac z4\phi_1(\frac z2)^2 & 0 & \phi_1(\frac z2) & 0 \\[2pt]
    \hline
    & \raisebox{-1pt}{$b_1(z)$} & \raisebox{-1pt}{$b_2(z)$} &
      \raisebox{-1pt}{$b_3(z)$} & \raisebox{-1pt}{$b_4(z)$}
  \end{array} \qquad
  \begin{aligned}
    b_1(z) &= {\jt\frac23}\phi_3(z) - {\jt\frac32}\phi_2(z) + \phi_1(z), \\
    b_2(z) &= b_3(z) = -{\jt\frac23}\phi_3(z) + \phi_2(z), \\
    b_4(z) &= {\jt\frac23}\phi_3(z) - {\jt\frac12}\phi_2(z),
    \end{aligned}
\end{equation}
where
\begin{equation}\label{eq:phi:k:def}
  \phi_k(z) = \left\{
  \begin{array}{cl}
    e^z, & k=0, \\
    k\int_0^1 \theta^{k-1} e^{(1-\theta)z}d\theta,
    & k\ge1
  \end{array}\right\}
  = \frac{k!}{z^k}\bigg( e^z - \sum_{l=0}^{k-1}\frac{z^l}{l!}\, \bigg),
  \quad k=0,1,2,\dots.
\end{equation}
The Butcher array for the Whalen-Brio-Moloney scheme is given in Table
2 of \cite{whalen:etd}.  Both of these schemes are explicit methods,
i.e.~$A(z)$ is strictly lower-triangular. Moreover, $c_r=\sum_{j=1}^s
A_{rj}(0)$.  It follows that $c_1=0$, $U_1=u_n$, and for $r\ge2$, the
upper limit of the sum in the formula for $U_r$ can be replaced by
$r-1$ so that only previously computed values of $\mc N_j$ are needed
to evaluate $U_r$. The only implicit ETD schemes we are aware of are
the high-order collocation methods of Chen and Wilkening
\cite{chen:wilkening,shkoller}.

Let $\mc V$ denote the space of real-valued functions defined on a
uniform $M_1\times M_2$ grid overlaid on the torus $\mbb T^2$ via
(\ref{eq:T2:discr}). When solving (\ref{eq:ssd:L})--(\ref{eq:ssd:N}),
the state space for (\ref{eq:etd:gen}) consists of state vectors
$u=(\tilde\eta,\tilde\varphi)\in\mc V^2$.  The nonlinear function $\mc
N(u)$ in (\ref{eq:ssd:N}) does not explicitly depend on time and is
computed using the pseudo-spectral approach described in
Section~\ref{sec:num}, i.e.~derivatives and the Hilbert transform are
computed in Fourier space while the quadratic nonlinearities in
(\ref{eq:ssd:N}) are evaluated pointwise on the grid. Evaluation of
$e^{c_rhL}u_n$, $a_{rj}(hL)\mc N_j$, $e^{hL}u_n$ and $b_r(hL)\mc N_r$
in (\ref{eq:ETD:step}) are all of the form $\psi(hL)v$ where $\psi(z)$
is an entire function involving the $\phi_k(z)$ functions and $v$ is a
state vector. Since the two-dimensional FFT diagonalizes $L$ (see
  below), the numerical instabilities discussed by Kassam and
Trefethen \cite{kassam} are avoided by using the series expansion of
(\ref{eq:phi:k:def}) for $|z|\le1$, which is
$\sum_{l=k}^\infty(k!/l!)z^{l-k}$. Replacing the upper limit by $k+19$
is sufficient to achieve double-precision accuracy.  There is no
catastrophic cancellation of digits since the leading terms of $e^z$
are eliminated analytically before evaluating the series
numerically. It is therefore not necessary to use the contour integral
approach advocated in \cite{kassam}. For $|z|>1$, one can just
evaluate $(k!/z^k)(e^z-\sum_{l=0}^{k-1}z^l/l!)$ as written, or use the
scaling and modified squaring algorithm of Skaflestad and Wright
\cite{skaflestad:09}; see also \cite{rocky:osc,chen:wilkening}.

Next we explain how to compute $\psi(hL)v$.  We temporarily allow
functions in $\mc V$ to be complex-valued and regard it as a complex
vector space. In the end, only real-valued functions in $\mc V$ will
actually arise in the calculation.  Let $\mc F$ denote the 2D FFT, and
consider the space $\hat{\mc V}=l^2(\mc J)$ of functions $\hat f$
taking values $\hat f_{j_1,j_2}$ on the discrete Fourier lattice
\begin{equation}
  \mc J = \{\bds j=(j_1,j_2)\,:\, -M_i/2+1\le j_i\le M_i/2\,,\, i=1,2\}.
\end{equation}
The operator $\mbb F=\opn{diag}(\mc F, \mc F)$ is an isomorphism of
the state space $\mc V^2$ onto $\hat{\mc V}^2$, and $L$ is
block-diagonalized by $\mbb F$.  In more detail $L=\mbb F^{-1}S\mbb
F$, where $S$ leaves invariant the two-dimensional Fourier subspaces
$(\hat{\mc V}^2)_{\bds j} = \opn{span}\{e_{\bds j}^+,e_{\bds j}^-
\}\subset\hat{\mc V}^2$, where $\bds j\in\mc J$ and
\begin{equation}
  e_{\bds j}^+=(\delta_{\bds j},0), \qquad\quad e_{\bds j}^-=(0,\delta_{\bds j}).
\end{equation}
Here $0\in\hat{\mc V}$ is the zero function ($0_{\bds l}=0$ for $\bds
  l\in\mc J$), and $\delta_{\bds j}\in\hat{\mc V}$ is a lattice
version of the Kronecker delta, i.e.~$(\delta_{\bds j})_{\bds l}=1$ if
$\bds l=\bds j$ and 0 if $\bds l\in\mc{J}\setminus\{\bds j\}$. The
restriction of $S$ to $(\hat{\mc V}^2)_{\bds j}$ has the following
matrix representation with respect to the basis $\{e_{\bds j}^+,
e_{\bds j}^-\}$:
\begin{equation}\label{eq:Sj:def}
  S_{\bds j} = \begin{pmatrix}
    0 & a_{\bds j} \\ -b_{\bds j} & 0
  \end{pmatrix}, \qquad
  a_{\bds j} = |j_1+j_2k|, \qquad
  b_{\bds j} =
    \begin{cases} 0, & \bds j=(0,0), \\
      g+a_j^2, & \text{otherwise}.
    \end{cases}
\end{equation}
When $\bds j=(0,0)$, $S_{\bds j}$ is the zero matrix, so it is
diagonalized by $E_{\bds j}=I_{2\times 2}$ with eigenvalues
$\lambda_{\bds j}^\pm=0$.  Otherwise we have $S_{\bds j} = E_{\bds
  j}\Lambda_{\bds j}E_{\bds j}^{-1}$ with
\begin{equation}\label{eq:Ej:def}
  \Lambda_{\bds j} = \begin{pmatrix} \lambda^+_{\bds j} & \\ &
    \lambda^-_{\bds j}
  \end{pmatrix}, \qquad
  E_{\bds j} = \begin{pmatrix} \sqrt{a_{\bds j}/b_{\bds j}} & \sqrt{a_{\bds j}/b_{\bds j}} \\
    i & -i \end{pmatrix}, \qquad
  E_{\bds j}^{-1} = \frac12\begin{pmatrix} \sqrt{b_{\bds j}/a_{\bds j}} & -i \\
    \sqrt{b_{\bds j}/a_{\bds j}} & i \end{pmatrix}
\end{equation}
and $\lambda^\pm_{\bds j} = \pm i\sqrt{a_{\bds j}b_{\bds j}}$.  Let
$\Lambda$ and $E$ be the operators on $\hat{\mc V}^2$ that leave the
subspaces $(\hat{\mc V}^2)_{\bds j}$ invariant and have matrix
representations $\Lambda_{\bds j}$ and $E_{\bds j}$ with respect to
$\{e_{\bds j}^+, e_{\bds j}^-\}$.
Then $Q=\mbb F^{-1}E$ diagonalizes $L$ and hence $\psi(hL)$:
\begin{equation}
  \begin{aligned}
    L &= Q\Lambda Q^{-1}, \\
    \psi(hL) &= Q\psi(h\Lambda)Q^{-1},
  \end{aligned} \qquad\quad
  \begin{aligned}
    \Lambda e_{\bds j}^\pm &= \lambda_j^\pm e_{\bds j}^\pm, \\
    \psi(h\Lambda)e_{\bds j}^\pm &= \psi(\lambda_j^\pm)e_{\bds j}^\pm.
  \end{aligned}
\end{equation}
The functions $\psi(z)$ that arise in (\ref{eq:ETD:step}) all have the
property that $\psi(\bar z) = \overline{\psi(z)}$. Let
$v\in\mc V^2$ be a real-valued state vector and
denote the intermediate steps in computing $\psi(hL)v$ as
\begin{equation}\label{eq:psi:hL:v}
  \xymatrixcolsep{2.5pc}
  \xymatrix{ \psi(hL)v  \;\;\;
    & \ar@{|->}[l]_-{\mbb F^{-1}} \;\;\; \hat y \;\;\;
    & \ar@{|->}[l]_-{E} \;\;\; \hat x \;\;\;
    & \ar@{|->}[l]_-{\psi(h\Lambda)} \;\;\; \hat w \;\;\;
    & \ar@{|->}[l]_-{E^{-1}} \;\;\; \hat v \;\;\;
    & \ar@{|->}[l]_-{\mbb F} \;\;\; v.
  }
\end{equation}
We denote the components of $\hat v$ in the $e_{\bds j}^\pm$ basis by
$\hat v_{\bds j}^\pm$, with similar notation for $\hat w$, $\hat x$,
and $\hat y$.  Since $v$ is real-valued and $a_{-\bds j}=a_{\bds j}$,
$b_{-\bds j}=b_{\bds j}$ in (\ref{eq:Sj:def}), we have
\begin{equation}\label{eq:v:E:lam}
  \hat v_{-\bds j} = \overline{\hat v_{\bds j}}, \qquad
  E_{-\bds j}=E_{\bds j}, \qquad
  \psi(h\lambda_{-\bds j}^\pm) =
  \psi(h\lambda_{\bds j}^\pm) = \psi\Big(\;\overline{h\lambda_{\bds j}^\mp}\;\Big)
  = \overline{\psi(h\lambda_{\bds j}^\mp)}.
\end{equation}
Inspecting the formulas in (\ref{eq:Ej:def}), we see that conjugating
the inputs of $E_{\bds j}^{-1}$ gives the conjugates of the outputs in
the opposite order. (When computing $b=Ax$, we refer to the components
  of $x$ as inputs to $A$ and those of $b$ as outputs.)  It then
follows from (\ref{eq:v:E:lam}) that
\begin{equation}
  \hat w_{-\bds j}^\pm = \overline{\hat w_{\bds j}^\mp}, \qquad
  \hat x_{-\bds j}^\pm = \psi(h\lambda_{-\bds j}^\pm)w_{-\bds j}^\pm
  = \overline{\psi(h\lambda_{\bds j}^\mp)w_{\bds j}^\mp} =
  \overline{\hat x_{\bds j}^\mp}.
\end{equation}
Conjugating and reversing the order of the inputs to $E_{\bds j}$
gives the conjugates of the outputs in the original order. Thus,
\begin{equation}\label{eq:yhat:sym}
  \hat y_{-\bds j}=\overline{\hat y_{\bds j}}
\end{equation}
and $\psi(hL)v$ is real. This also justifies using the `r2c'
version of the two-dimensional FFT to compute $\hat v=\mbb Fv$, which
only returns values of $v_{\bds j}$ with $\bds j=(j_1,j_2)\in\mc J$
and $j_1\ge0$. We then only compute $\hat w$, $\hat x$ and $\hat y$
for these values of $\bds j$. The missing entries are assured to
satisfy (\ref{eq:yhat:sym}), which is the assumption needed to
apply the `c2r' version of the inverse FFT to obtain
$\psi(hL)v=\mbb F^{-1}\hat y$.

The most expensive steps of evaluating $\psi(hL)v$ are the FFTs in
$\mbb F$ and $\mbb F^{-1}$. To implement (\ref{eq:ETD:step}), one has
to apply $\mbb F$ to $u_n$ and $\mc N_1,\dots,\mc N_s$ and apply $\mbb
F^{-1}$ to obtain $U_2,\dots,U_s$ and $u_{n+1}$. Since $\mbb F$ and
$\mbb F^{-1}$ each involve 2 FFT's, the functional calculus steps of
(\ref{eq:ETD:step}) involve a total of $4s+2$ FFT's. Meanwhile,
evaluating $\mc N_1,\dots,\mc N_s$ using the pseudo-spectral method to
compute derivatives and Hilbert transforms in (\ref{eq:ssd:N})
involves $10s$ FFTs. These would have to be computed using a
Runge-Kutta method anyway, so the cost of an $s$-stage ETD method is
approximately 40\% higher than that of an $s$-stage RK method with the
same stepsize $h$. But as seen in Figure~\ref{stepper_plot}, significantly
larger steps can often be taken with the ETD method for a given
accuracy goal, making the ETD method more efficient in spite of the
additional cost per step.

We remark that we use the formulas for $a_{\bds j}$ and $b_{\bds j}$
in (\ref{eq:Sj:def}) even for the Nyquist modes $\bds j=(j_1,j_2)$
with $j_1=M_1/2$ or $j_2=M_2/2$. We avoid setting $a_{\bds j}=0$ and
$b_{\bds j}=g$ for these modes, which would have been consistent
with our treatment of $H$, $\pa_{\alpha}^2$ and $P$ in the definition
of $L$ in (\ref{eq:ssd:L}), as it would lead to a Jordan block in the
diagonalization of $L$.  It makes little difference since the Nyquist
modes are intended to remain close to roundoff-level values throughout
the computation, and in fact are set to zero at the end of each
timestep by the filter (\ref{eq:filter}). In general, if a Jordan
block arises in the diagonalization of $L$ and the corresponding
eigenspace contains ``low-frequency'' modes that have to be resolved
to get an accurate result, one can transfer a term from $L$ to $\mc N$
in the decomposition (\ref{eq:etd:gen}) so that the modified $L$ is
diagonalizable.  This technique is demonstrated in \cite{shkoller}.


\section{The spatially quasi-periodic water wave problem in 3D}
\label{sec:3d}


In this section we briefly outline how to formulate the equations of
motion describing the evolution of water waves at the surface of a
three-dimensional fluid with spatially quasi-periodic boundary
conditions.  Since the conformal mapping framework does not generalize
to this setting, we will only discuss the graph-based formulation in
physical space.  Let $\Phi(x_1,x_2,y,t)$ denote the velocity potential
in the fluid.  The surface variables on $\mbb R^2$ that describe the
state of the system are
\begin{equation}
  \eta(x_1,x_2,t), \qquad \varphi(x_1,x_2,t)=\Phi(x_1,x_2,\eta(x_1,x_2,t),t).
\end{equation}
The equations of motion governing their evolution may be written
\cite{lannes:05}
\begin{align}
  \label{eq:eta:t:3d}
  \eta_t &= G(\eta)\varphi, \\
  \label{eq:phi:t:3d}
  \varphi_t &= \frac12\left(
    \frac{\big( G(\eta)\varphi + \nabla_{\bds x}\eta \cdot
        \nabla_{\bds x}\varphi\big)^2}{1+\big|\nabla_{\bds x}\eta\big|^2}
    - \big|\nabla_{\bds x}\varphi\big|^2\right) - g\eta + \tau\kappa
  + C(t)
\end{align}
where $C(t)$ is an arbitrary function of time, $\nabla_{\bds
  x}=\big(\pa_{x_1},\pa_{x_2}\big)^T$, and
\begin{equation}\label{eq:kappa:3d}
  \kappa=\nabla_{\bds x}\cdot\left(\frac{\nabla_{\bds x}\eta}{
      \sqrt{1+|\nabla_{\bds x}\eta|^2}}\right)
\end{equation}
is the mean curvature. We have also introduced the
Dirichlet-Neumann operator \cite{CraigSulem},
\begin{equation}\label{eq:G:def:3d}
  G(\eta)\varphi = \Phi_y - \eta_{x_1}\Phi_{x_1} - \eta_{x_2}\Phi_{x_2},
\end{equation}
where $\Phi(x_1,x_2,y,t)$ is the solution of
\begin{equation}
  \begin{alignedat}{2}
    \big( \pa_{x_1}^2 + \pa_{x_2}^2 + \pa_y^2\big) \Phi &= 0, \qquad
    & -\infty &<y < \eta(\bds x,t), \quad \bds x\in\mbb R^2, \\
    \Phi &= \varphi, &       y &= \eta(\bds x,t), \\
    \Phi_y &\rightarrow 0, & y &\rightarrow-\infty.
  \end{alignedat}
\end{equation}
Note that $t$ could be dropped from the notation in $\varphi$, $\eta$
and $\Phi$ when defining $G$ as time is frozen when solving the
auxiliary problem of reconstructing the velocity potential $\Phi$ in
the fluid from its boundary value $\varphi$ on the free surface $\eta$
and computing the scaled normal derivative (\ref{eq:G:def:3d}). Using
$\Phi_{x_i}=\varphi_{x_i} - \eta_{x_i}\Phi_y$ in (\ref{eq:G:def:3d}),
we see that (\ref{eq:phi:t:3d}) is equivalent to
\begin{equation}\label{eq:phi:t:3d:b}
  \varphi_t = \frac12\bigg(
    \Big(1 + \big|\nabla_{\bds x}\eta\big|^2\Big)\Phi_y^2 -
    \big|\nabla_{\bds x}\varphi \big|^2 \bigg) - g\eta + \tau\kappa + C(t),
\end{equation}
which can also be derived easily from the 3D analog of
(\ref{Bernoulli}).

Following the definitions in \cite{dynnikov2005topology},
we consider quasi-periodic functions on $\mbb R^2$ with
$d$ quasi-periods, which have the form
\begin{equation*}
  u(\bds{x}) = \tilde{u}(\bds{K}\bds{x})
  = \sum_{\bds j\in\mbb Z^d}\hat u_{\bds j}e^{i\bds j^T\bds K\bds x},  \qquad 
  \bds{x} = (x_1; x_2) \in\mathbb{R}^2, \qquad
  \bds{K} \in \text{Mat}_{d\times 2}(\mathbb R).
\end{equation*}
Here $\tilde{u}:\mbb R^d\to\mbb R$ is real analytic and periodic (so
  well-defined on $\mbb T^d$), $\hat u_{\bds j}$ are its Fourier
modes, and a semicolon separates entries of a column vector. The rows of $\bds K$
may be assumed to be linearly independent over $\mbb Z$ since
otherwise a new function $\tilde v:\mathbb{T}^{d-1} \to \mathbb{R}$
and matrix $\bds L\in \text{Mat}_{(d-1)\times 2}(\mathbb R)$ can be
constructed so that $u(\bds x)=\tilde v(\bds L\bds x)$. Indeed, if the
rows of $\bds K$ are linearly dependent over $\mbb Z$, there is a
unimodular integer matrix $\bds J$ such that $\bds J\bds K=[\bds
  L;\bds 0]$, i.e.~the last row of $\bds J\bds K$ contains only zeros
and the first $d-1$ rows define $\bds L$. One can then define $\hat
v_{\bds l}=\sum_{l_d\in\mbb Z} \hat u_{\bds j(\bds l,l_d)}$ and
confirm that $u(\bds x)=\sum_{\bds l\in\mbb Z^{d-1}}\hat v_{\bds
  l}e^{i\bds l^T\bds L\bds x}$, where $\bds j(\bds l,l_d) = \bds
J^T(l_1;\dots;l_{d-1};l_d)$ for $\bds l\in\mbb Z^{d-1}$ and $l_d\in\mbb Z$. One
also finds that $\tilde v(y)= \tilde u\big(J^{-1}(y;0)\big)$
for $y\in \mbb T^{d-1}$.  We do not know a reference for these
calculations, but they are straightforward. The procedure can be
repeated until the minimal $d$ is found.  For the wave to be genuinely
two-dimensional and quasi-periodic, we require $\bds K$ to have full
rank and $d\ge3$.  For example, one choice of $\bds{K}$ when $d = 3$
is
\begin{equation*}
  \bds K = \begin{pmatrix}
    1&0\\
    0&1\\
    \sqrt{2}&\sqrt{3}
  \end{pmatrix}.
\end{equation*}

Just as in Remark~\ref{rmk:torus}, substitution of the quasi-periodic
functions
\begin{equation}\label{eq:eta:phi:3d}
  \eta(\bds{x}, t) = \tilde{\eta}(\bds{K}\bds{x}, t), \qquad\quad
  \varphi(\bds{x}, t) = \tilde{\varphi}(\bds{K}\bds{x}, t)
\end{equation}
into (\ref{eq:eta:t:3d})--(\ref{eq:kappa:3d}) gives
evolution equations for $\tilde\eta(\bds{\tilde x},t)$ and
$\tilde\varphi(\bds{\tilde x},t)$ on the torus $\mbb T^d$.
One just has to replace $\nabla_{\bds x}$ by
$(\bds k_1\cdot\nabla_{\bds{\tilde x}};\bds k_2\cdot\nabla_{\bds{\tilde x}})$,
where $\bds k_1,\bds k_2\in\mbb R^d$ are the columns of $\bds K$
and $\bds{\tilde x}=(\tilde x_1;\dots;\tilde x_d)\in\mbb T^d$.
This amounts to using the chain rule,
$D_{\bds x}\eta=D_{\bds{\tilde x}}\tilde\eta\cdot \bds K$, where
$D_{\bds x}=(\pa_{x_1},\pa_{x_2})=(\nabla_{\bds x})^T$. It is also
necessary to define a quasi-periodic Dirichlet-Neumann operator
via
\begin{equation}
  \tilde G(\tilde\eta)\tilde\varphi = \widetilde{G(\eta)\varphi}.
\end{equation}
Given $\tilde\eta$ and $\tilde\varphi$ in (\ref{eq:eta:phi:3d}) with
$t$ fixed, one would construct the solution
$\tilde\Phi(\bds{\tilde x},y,t)$ of
\begin{equation}\label{eq:laplace:3d:Q}
  \begin{alignedat}{2}
    (\bds{k}_1\cdot\nabla_{\tilde{\bds{x}}})^2 \tilde{\Phi} 
    + (\bds{k}_2\cdot\nabla_{\tilde{\bds{x}}})^2\tilde{\Phi}
    + (\partial_y)^2\tilde{\Phi} &= 0,  
    \qquad &-\infty &< y < \tilde{\eta}(\tilde{\bds{x}}, t), \quad
    \bds{\tilde x}\in\mbb R^d, \\
    \tilde\Phi (\tilde{\bds{x}}, \tilde{\eta}(\tilde{\bds{x}}, t),t)
    &= \tilde\varphi(\tilde{\bds{x}}, t), 
    \qquad & y &= \tilde\eta(\tilde{\bds{x}}, t), \\
    \partial_y \tilde{\Phi} &\to 0, 
    \qquad & y &\to -\infty
  \end{alignedat}
\end{equation}
and then evaluate
\begin{equation*}
  \tilde G(\tilde{\eta})\tilde{\varphi} = \Big[\partial_y\tilde{\Phi}
    - (\bds{k}_1\cdot\nabla_{\tilde{\bds{x}}}\tilde{\eta})
    (\bds{k}_1\cdot\nabla_{\tilde{\bds{x}}}\tilde{\Phi} )
    -(\bds{k}_2\cdot\nabla_{\tilde{\bds{x}}}\tilde{\eta})
    (\bds{k}_2\cdot\nabla_{\tilde{\bds{x}}}\tilde{\Phi} )
  \Big]_{y = \tilde{\eta}}.
\end{equation*}
A calculation similar to the one for the periodic case
\cite{zakharov1968stability} shows that the torus version
of (\ref{eq:eta:t:3d})--(\ref{eq:phi:t:3d}) is a Hamiltonian
system with energy
\begin{equation}
  E = \frac1{(2\pi)^d}\int_{\mbb T^d} \frac{1}{2}
  \tilde{\varphi}\tilde G(\tilde{\eta})\tilde{\varphi} + \frac{1}{2} g
  \tilde{\eta}^2 + \tau \left(\sqrt{1+(\bds{k}_1\cdot
        \nabla_{\tilde{\bds{x}}}\tilde{\eta})^2 +(\bds{k}_2\cdot
        \nabla_{\tilde{\bds{x}}}\tilde{\eta})^2} - 1\right)
  d\tilde{\bds{x}},
\end{equation}
and $\tilde\eta$ and $\tilde\varphi$ are conjugate variables.

Whereas the conformal mapping approach of Section~\ref{sec:gov:conf}
employs a quasi-periodic Hilbert transform to efficiently compute the
Dirichlet-Neumann operator for a 2D fluid, it remains an open problem
to devise and implement an efficient method for solving
(\ref{eq:laplace:3d:Q}) for a 3D fluid with quasi-periodic boundary
conditions.  In a finite-depth variant of the problem, the finite
element approach of Wilkening and Rycroft
\cite{rycroft2013computation} and the Transformed Field Expansion
method of Nicholls and Reitich
\cite{nicholls:reitich:01,nicholls:reitich:06,qadeer} are candidate
approaches that have been used successfully for periodic boundary
conditions. Both approaches could work in principle for quasi-periodic
boundary conditions but will suffer from the curse of dimensionality
as the $(d+1)$-dimensional region
\begin{equation*}
  \Omega_{\tilde{\eta},h,t} = \big\{ (\bds{\tilde x};y)\;:\;
  -h < y < \tilde\eta(\bds{\tilde x},t)\;,\; \bds{\tilde x}\in\mbb T^d
  \big\}
\end{equation*}
has to be discretized, where $h$ is the fluid depth. We do not know
what to expect for the condition number of a linear system that
discretizes (\ref{eq:laplace:3d:Q}), which is elliptic with respect to
$(\bds x;y)$ on quasi-periodic slices through
$\Omega_{\tilde{\eta},h,t}$ but not with respect to $(\bds{\tilde
    x};y)$ on $\Omega_{\tilde{\eta},h,t}$. The infinite depth case can
often be dealt with by introducing a transparent boundary condition
along a fictitious interface at some depth $y=-h$ below the free
surface, discretizing the fluid region $\Omega_{\tilde{\eta},h,t}$
above the interface, and using series expansions in the unbounded
region below the interface \cite{nicholls:reitich:06}. We have not
worked out the details in a quasi-periodic setting.

A final option is
to avoid the Dirichlet-Neumann operator altogether by using a weakly
nonlinear model water wave equation. One could introduce torus versions of
the equations of motion to obtain genuine quasi-periodic solutions
rather than solving a system of coupled single-mode nonlinear
Schr\"odinger equations \cite{bridges2001,ablowitz2015interacting}.
The curse of dimensionality will also be a challenge for weakly
nonlinear models with this torus approach, especially if larger
values of $d$ are considered.

\bibliographystyle{abbrv}


\end{document}